\documentclass[sigconf]{aamas} 
\usepackage[noend]{algpseudocode}
\usepackage{amsthm}
\usepackage{graphicx}
\usepackage{subcaption}
\usepackage{amsmath}
\usepackage{bbm}
\usepackage{multicol}

\usepackage{xcolor}

\DeclareMathOperator*{\argmax}{arg\,max}

\usepackage{xcolor}
\usepackage[linesnumbered,ruled,vlined]{algorithm2e}
\SetKwIF{If}{ElseIf}{Else}{if}{}{else if}{else}{end if}%

\SetCommentSty{mycommfont}
\SetKwInput{KwInput}{Input}                % Set the Input
\SetKwInput{KwOutput}{Output}              % set the Output
\SetKwInput{KwInitialise}{Initialise}      % set the Output

\algnewcommand{\algorithmicgoto}{\textbf{go to}}%
\algnewcommand{\Goto}[1]{\algorithmicgoto~\ref{#1}}%

\newcommand{\Er}{\mathbf{E}}

\newcommand{\Ac}{\ensuremath{\mathcal{A}}}

\newcommand{\Gc}{\ensuremath{\mathcal{G}}}

\newcommand{\Nc}{\ensuremath{\mathcal{N}}}

\newcommand{\Uc}{\ensuremath{\mathcal{U}}}

\newcommand{\Rbb}{\ensuremath{\mathbb{R}}}
\newcommand{\Ibb}{\ensuremath{\mathbb{I}}}

\newcommand{\lb}{\ensuremath{\left(}}
\newcommand{\rb}{\ensuremath{\right)}}
\newcommand{\lbr}{\ensuremath{\left\{}}
\newcommand{\rbr}{\ensuremath{\right\}}}
\newcommand{\lsb}{\ensuremath{\left[}}
\newcommand{\rsb}{\ensuremath{\right]}}

\newtheorem{theorem}{Theorem}
\newtheorem{lemma}{Lemma}
\newtheorem{definition}{Definition}

\usepackage{balance} % for balancing columns on the final page

\setcopyright{rightsretained}
\acmConference[GAIW'23]{Appears at the 5th Games, Agents, and Incentives Workshop (GAIW 2023). Held as part of the Workshops at the 21st International Conference on Autonomous Agents and Multiagent Systems.}{May-June 2023}{London, England}{Abramowitz, Aziz, Ceppi, Dickerson, Hosseini, Lev, Mattei, Tsang, Wąs, Zick (Chairs)} 
\copyrightyear{2023}
\acmYear{2023}
\acmDOI{}
\acmPrice{}
\acmISBN{}

\acmSubmissionID{???}

\title[AAMAS-2023 Formatting Instructions]{Social Optimum Equilibrium Selection for Distributed Multi-Agent Optimization}

\author{Duong Nguyen}
\affiliation{
  \institution{The University of Adelaide}
  \city{Adelaide}
  \country{Australia}}
\email{duong.nguyen@adelaide.edu.au}

\author{Langford White}
\affiliation{
  \institution{The University of Adelaide}
  \city{Adelaide}
  \country{Australia}}
\email{langford.white@adelaide.edu.au}

\author{Hung Nguyen}
\affiliation{
  \institution{The University of Adelaide}
  \city{Adelaide}
  \country{Australia}}
\email{hung.nguyen@adelaide.edu.au}

\begin{abstract}
We study the open question of how players learn to play a social optimum pure-strategy Nash equilibrium (PSNE) through repeated interactions in general-sum coordination games. A social optimum of a game is the stable Pareto-optimal state that provides a maximum return in the sum of all players' payoffs (social welfare) and always exists. We consider finite repeated games where each player only has access to its own utility (or payoff) function but is able to exchange information with other players. We develop a novel regret matching (RM) based algorithm for computing an efficient PSNE solution that could approach a desired Pareto-optimal outcome yielding the highest social welfare among all the attainable equilibria in the long run. Our proposed learning procedure follows the regret minimization framework but extends it in three major ways: (1) agents use global, instead of local,  utility for calculating regrets, (2) each agent maintains a small and diminishing exploration probability in order to explore various PSNEs, and (3) agents stay with the actions that achieve the best global utility thus far, regardless of regrets. We prove that these three extensions enable the algorithm to select the stable social optimum equilibrium instead of converging to an arbitrary or cyclic equilibrium as in the conventional RM approach. We demonstrate the effectiveness of our approach through a set of applications in multi-agent distributed control, including a large-scale resource allocation game and a hard combinatorial task assignment problem for which no efficient (polynomial) solution exists. 
\end{abstract}


\newcommand{\BibTeX}{\rm B\kern-.05em{\sc i\kern-.025em b}\kern-.08em\TeX}

\begin{document}

%%% The following commands remove the headers in your paper. For final 
%%% papers, these will be inserted during the pagination process.

\pagestyle{fancy}
\fancyhead{}

%%% The next command prints the information defined in the preamble.

\maketitle 

%%%%%%%%%%%%%%%%%%%%%%%%%%%%%%%%%%%%%%%%%%%%%%%%%%%%%%%%%%%%%%%%%%%%%%%%

\section{Introduction}

The ability to coordinate large and complex multi-agent systems (MAS) to achieve a system-level goal is a major engineering challenge. In particular, problems of scaling, observability, and constraints on computational resources make it impractical to apply centralized control methods~\cite{konda2021mission}. Game theoretic  multi-agent decision-making has recently been proposed as a  distributed paradigm to tackle many of the associated challenges~\cite{paccagnan2022utility}. The key questions in these game-theoretic distributed control systems are the game-based strategies to ensure the collective behavior of all agents converges to a stable optimum equilibrium with respect to a performance metric of interest, such as a social optimum state that provides maximum return in the sum of all players' utilities (social welfare). Thus, the utility functions of the players and the global utility have to be designed in order that even if the players are self-interested (i.e., they maximize their private, local utility functions), they may indirectly cooperate to achieve the global objective. One desired outcome in these systems is to achieve a Pareto optimal equilibrium, where no player can benefit from unilaterally moving away from this outcome, as it will necessarily result in a loss to someone else~\cite{konda2021mission}. %In other words, a Pareto-optimal outcome cannot be improved upon without disadvantaging at least one player and thus it is a natural goal for a MAS to drive agents to.

Most current game-theoretic distributed control solutions focus on achieving a Nash equilibrium (NE) play~\cite{paccagnan2022utility}, but not the Pareto optimal solutions needed in distributed MAS systems. The main reason is that it is difficult to achieve mutual cooperation due to the lack of individual control in large games.  Unfortunately, finding a NE or even an approximate NE in a multiplayer game is computationally intractable~\cite{berg2017exclusion}. Computing an exact desired equilibrium among the set of approachable Nash equilibria is even harder~\cite{barman2015finding}. %These algorithms cannot be used for finding Pareto optimal solutions in MAS - our focus in this paper.

Our solution builds on the general principles of regret matching (RM)~\cite{hart2000simple}. RM-based algorithms are known to be efficient algorithms for finding equilibria in large games~\cite{zhang2022equilibrium} with nice properties such as  guaranteed convergence to the set of coarse correlated equilibria (CCE) in finite games~\cite{hart2000simple}, and pure Nash equilibria in socially concave games~\cite{hart2015markets}. In some categories of games, such as potential games or weakly acyclic games, RM can be shown to converge to the NE~\cite{marden2007regret, saran2012regret}. Despite their many positive properties, current RM algorithms, however, are not designed to efficiently find a stable Pareto optimal solution~\cite{zhang2022equilibrium}.

\noindent\textbf{Our Contributions.} In this paper, we develop a new RM-based algorithm
%the standard Regret Matching~\cite{hart2000simple}, which is known as the state-of-the-art distributed approach for computing equilibria in large-scale games~\cite{feng2021convergence}, 
to compute social Pareto optimal equilibria in large-scale distributed control games for MAS. We make three major modifications to the standard RM in order to guarantee convergence and efficiency.
%and propose several ways to address its shortcomings and improve the efficiency of convergence points. %Different from the original approach, which originally proposed for solving non-cooperative games, our focus is on the class of coordination games that is directly associated with multi-agent distributed optimization problems. More specifically, our work suggest that RM based framework can be applied for playing a repeated coordination game and obtain a stable social optimum outcome by designing an appropriate utility function. 
More specifically,
\begin{itemize}
    \item We develop a new RM-based algorithm for finding stable and Pareto-optimal Nash equilibria in general-sum coordination games. Our algorithm can approximate a social optimum equilibrium
    %, which is the stable Pareto-optimal outcome that yields the maximal return in the sum of all players' individual utilities. 
   by using a global utility (i.e, the sum of all players' local utilities) instead of local utilities, in computing the players' regrets for each player. With this utility function, the game is a finite potential game where a stable set of pure-strategy Nash equilibria is guaranteed to exist~\cite{monderer1996potential}. 
    \item We prove a general result that the empirical distribution of joint play is guaranteed to converge to an optimal pure-strategy Nash equilibrium (PSNE) in games where each player uses a local utility function that is appropriately aligned with the global objective, i.e., such that maximizing the local utilities corresponds to maximizing the global utility. We further show  that a joint strategy that maximizes the global utility is also a desired Pareto-optimal PSNE outcome for which there is no other game outcome in which every player is better off. %, which can not be improved further in term of maximizing collective utility of all players.
    \item We introduce an adaptive monitoring threshold to speed up and control the convergence point of the learning procedure. We demonstrate that our algorithm not only successfully approaches a social optimum outcome but also empirically converges faster compared to the standard RM method~\cite{hart2000simple}. %than a state-of-the-art regret minimization based method~\cite{}.
\end{itemize}

%%%%%%%%%%%%%%%%%%%%%%%%%%%%%%%%%%%%%%%%%%%%%%%%%%%%%%%%%%%%%%%%%%%%%%%%

\section{Background and Related Works}
\subsection{Game-Theoretic Distributed Control}

The game-theoretic approach to distributed optimization in multi-agent systems generally requires: (1) assigning local utility functions for players that align with the global objective of the system, and (2) designing an iterative decision-making rule to learn an optimal game stage outcome. Many distributed optimization tasks, such as distributed resource allocation where players without a central correlation mechanism can cooperate to achieve a globally optimal resource distribution, can be modeled as potential games~\cite{marden2015game} under this framework. The search for an optimal solution in these approaches can be seen as a task of finding NE of the game. Most of the learning algorithms for potential games guarantee convergence to a pure Nash equilibrium~\cite{marden2009cooperative}. 

In these methods, each player strives to optimize its local utility function, which is associated with the global objective function. The local utility function of each player often not only depends on its action but also the actions chosen by other players. Thus, the coordination problem can be interpreted as a distributed optimization problem, in which iterative game-based learning algorithms can be used to search for the NE of the coordination game. Efficient tractable algorithms for finding the optimal NEs in these games are still an open question~\cite{li2022role}. 

\subsection{The Problem of Computing Arbitrary NE}
Although there are proven methods to find a NE outcome in repeated normal-form games, such as Fictitious Play~\cite{brown1951iterative} or tree search-based methods~\cite{berg2017exclusion}, if there exist multiple Nash equilibria, there is no prescription for deciding which one a player should play~\cite{feng2021convergence}. Thus, if players choose their strategies from different equilibria, then the result is not necessarily an equilibrium solution. Therefore, the critical thing in these types of situations is determining mechanisms that allow players to coordinate on one specific equilibrium among all attainable equilibria. This is a challenge yet to be addressed in multi-agent learning in games~\cite{raducha2022coordination}. 

Another major issue of reaching an arbitrary NE is that, in general, it does not correspond to an optimal outcome~\cite{vetta2002nash,roughgarden2015intrinsic}. That is, for a given game, even when payoffs are better if players coordinate on one specific equilibrium, using existing  algorithms players  tend to choose  options leading to an equilibrium that is sub-optimal. This issue is related to a research thread within algorithmic game theory that focuses on analyzing the inefficiency of Nash equilibria according to two worst-case
measures, the price of anarchy and the price of stability~\cite{koutsoupias1999worst}. Note however that these studies focus on analysis and not synthesis of game-based algorithms.

%===================================
\subsection{Finding Pareto Optimal Solution}
%\HN{Talk about different approaches to find Pareto optimal point. Why RM is a leading candidate? Why do you want to build on RM here?}
One feasible way to achieve a Pareto optimal equilibrium is to use the linear programming (LP) method~\cite{berri2016correlated}. However, computational complexity that grows exponentially with the number of players and actions makes this approach infeasible for large-size problems. An alternative option is to use a tree-search based method, such as~\cite{berg2017exclusion}. Again, the run time that grows rapidly with the game size is the main drawback of applying this method in practice. There is currently no feasible algorithm for reaching Pareto optimal equilibrium for games with large state space~\cite{chandan2021tractable}.

\subsection{Our Key Innovation} 
Marden \emph{et al}~\cite{marden2009cooperative} demonstrated that one could use a regret minimization based method to reach a pure Nash solution for the cooperative control problem in a specific class of game setting, i.e., potential games and its generalized form known as weakly acyclic games. The work~\cite{nguyen2019adaptive} demonstrated that by designing a suitable local utility function that takes into account all agents' behaviors, a regret-based algorithm can also be used to obtain cooperative behaviors among the agents. The most recent work~\cite{zhang2022equilibrium} proposed an extended RM aiming to compute CE in general-sum games that offer higher overall social welfare compared to equilibria found by the original RM approach. 

None of these works investigate equilibrium selection techniques that can achieve a specific outcome. Equilibrium selection, which is to compute an exact equilibrium that achieves a non-trivial approximation to the optimal welfare, is hard~\cite{barman2015finding,raducha2022coordination}. Our main innovation in this paper is a novel equilibrium selection method that achieves a Pareto-optimal solution for coordination games.%\footnote{Coordination games model situations in which all players can obtain mutually optimal outcomes, but only by playing corresponding strategies.}. %Our method works by blah blah \HN{we need to spell out how we we solve the problem at the high-level}. %(also known as common-interest games)

%%%%%%%%%%%%%%%%%%%%%%%%%%%%%%%%%%%%%%%%%%%%%%%%%%%%%%%%%%%%%%%%%%%%%%%%

\section{Problem Definition}
A (finite) repeated normal-form game can be defined as a tuple $\Gc = \big(\Nc,\{\Ac_i\}_{i\in\Nc},\{u_i\}_{i\in\Nc}\big)$, where $\Nc = \{1, 2, \dots, |\Nc|\}$\footnote{$|\cdot|$ denotes the cardinality of a set.} is a finite  set of players, $\Ac_i = \{a_1, a_2, \dots, a_{|\Ac_i|}\}$ is the set of pure actions for each player $i$, and $u_i: \Ac \rightarrow \mathbb{R}$ specifies a payoff function for player $i\in \Nc$, where $\Ac = \prod_{j\in\Nc} \Ac_j$ is the set of joint pure actions for all players. Player $i$ decides its action according to a strategy $\pi_i\in \Delta(\Ac_i)$\footnote{$\Delta(\cdot)$ denotes the probability distribution over a set.}, which can be a pure strategy (choosing to take one action with prob. of $1$) or a mixed strategy (selecting an action according to a prob. distribution over the set of its pure actions).    

Let $\pi = (\pi_i,\pi_{-i}) \in \Delta(\Ac)$ be a strategy profile, a set consisting of one strategy for each player, where player $i$ follows its strategy $\pi_i \in \Delta(\Ac_i)$ and the remaining players follow their strategy $\pi_{-i}  \in \Delta(\Ac_{-i})$\footnote{$\Ac_{-i} = \Ac \setminus \Ac_i$ is the action set of all players excluding $i$.}. Then we can calculate the expected payoff of player $i$ as a function of $\pi$ as follows:
\begin{equation} \label{eq:payoff}
    u_i(\pi_i,\pi_{-i}) = \sum\nolimits_{a_i\in \Ac_i}\ \sum\nolimits_{a_{-i}\in \Ac_{-i}}\ \pi(a_i,a_{-i})\ u_i(a_i,a_{-i})\ ,
\end{equation}
where $(a_i,a_{-i})$ denotes a pure action profile in which player $i$ chooses action $a^{\prime}_i$ and the other remaining players select $a_{-i}$; and $\pi(a_i,a_{-i})=\pi_i(a_i)\ \pi_{-i}(a_{-i})$ is the probability that the action profile $(a_i,a_{-i})$ is played. %under the strategy profile $\pi$. 

\subsection{Equilibrium and Optimality Concepts}
An equilibrium solution to the game $\Gc$ is a specific strategy profile $\pi^* = (\pi_i^*,\pi_{-i}^*)$ that leaves all players no incentive to deviate unilaterally. In non-cooperative games, NE is well-known core concept defined as when there is a strategy profile such that no player can increase its expected payoff by deviating from its current strategy, given that the other players do not change their strategies.  

\begin{definition}[Nash Equilibrium]\label{def:PSNE}
A mixed strategy profile $\pi^* \in \Delta(\Ac)$ is a mixed NE if for all $i\in \Nc$ and all $\pi_i \in \Ac_i$ %\neq \pi_i$
% $$u_i(\pi^*_i,\pi^*_{-i}) = \max_{\pi_i \in \Delta(\Ac_i)} u_i(\pi_i,\pi^*_{-i}) \ .$$
\begin{equation} \label{eq:NE}
    u_i(\pi^*_i,\pi^*_{-i}) \geq u_i(\pi_i,\pi^*_{-i})\ .
\end{equation}
A pure-strategy Nash equilibrium (PSNE) is an action profile $a^*=(a^*_i,a^*_{-i}) \in \Ac$ that satisfies the same conditions
\begin{equation} \label{eq:PSNE}
u_i(a^*_i,a^*_{-i}) \geq u_i(a_i,a^*_{-i}) \ .
\end{equation}
\end{definition}

It is difficult to compute a NE, and a NE solution often does not provide high expected payoffs to the players. Coarse correlated equilibrium, introduced by Aumann~\cite{aumann1987correlated}, is an optimality concept that generalizes NE and is easier to compute than NE. Under a CCE, no player can benefit on average by playing any fixed action. It can be seen that the CCE distributions that are independent across players are precisely the NE of the game.   

\begin{definition}[Coarse Correlated Equilibrium]\label{def:CCE}
A mixed strategy profile $\pi^* \in \Delta(\Ac)$ is a CCE if $\forall i\in \Nc$ and $\forall a^{\prime}_i\in\Ac_i$%, it holds that
\begin{equation} \label{eq:CCE}
u_i(\pi^*_i,\pi^*_{-i}) \geq \sum\nolimits_{a_i\in \Ac_i} \sum\nolimits_{a_{-i}\in \Ac_{-i}} \pi^*(a_i,a_{-i})\ u_i(a^{\prime}_i,a_{-i}) \ . 
\end{equation}
\end{definition}

At a CCE, all players are willing to commit to follow the CCE strategy given that all the others also choose to commit. That is, if a single player decides not to commit to follow the CCE strategy, it experiences a lower expected payoff.

For optimality, the global performance of an action profile can be measured through a global objective (also known as social welfare) function $W: \Ac\rightarrow \mathbb{R}$, defined as 
$$W(a_i,a_{-i}) = \sum\nolimits_{i\in\Nc} u_i(a_i,a_{-i})\ .$$ 
An action profile that maximizes the social welfare is said to be Pareto-optimal since it is impossible to increase the local payoff of a player without decreasing the payoff of another.
\begin{definition}[Pareto Optimality]
A joint action profile $(a^*_i,a^*_{-i})$ is Pareto-optimal if $\forall i\in\Nc$ and $\forall (a_i,a_{-i})\in\Ac$
\begin{equation} \label{eq:Pareto}
    W(a^*_i,a^*_{-i}) = \max_{(a_i,a_{-i})\in\Ac}\ \sum\nolimits_{i\in\Nc} u_i(a_i,a_{-i})\ .
\end{equation}
\end{definition}
In words, an action profile is Pareto-optimal if there is no other profile that would increase payoff of any player without reducing payoff of at least one player. Pareto optimality provides an indication of whether or not an outcome is good from the perspective of social welfare. Pareto optimality can also extends to distribution on action profiles. A strategy profile is Pareto-optimal if there is no other distribution that provides one agent with strictly greater expected payoff without giving another agent strictly less.  

\subsection{Examples of A Social Optimum Outcome}
%For a game with a single NE, it is reasonable to predict that players will play the strategies in this equilibrium. Most natural games, however, often have more than one NE, and thus it becomes difficult to predict how players will behave in the game and how one unique NE can be selected. %(i.e., coordination games)

We use a simple $2\times2$ resource allocation game in wireless networks to illustrate the concepts introduced so far in this paper. In this example, there are two players wishing to access one of two wireless service resources. The set of actions is denote by $\Ac = \{(j,k):j,k = R\#1,R\#2\}$ where $(j,k)$ means that player 1 chooses resource $j$ and player 2 chooses resource $k$. Each player has a different physical (PHY) rate on each resource (measured as Mbps as illustrated on the link from a player to a resource). The PHY rate of a player is the maximum obtainable throughput when it connected alone to a resource. When they both connect to the same resource, they have to share the resource and thus obtain a lower throughput compared to their maximum physical rates. Practical implementation of these schemes using is currently available~\cite{7348945}. The utility (payoff) functions are the real throughput observed by each player. The game outcomes for every possible joint action combination is summarized in Figure~\ref{fig:2x2resoucegame}.

\begin{figure}[!h]
 \centering
 \includegraphics[width=0.68\linewidth]{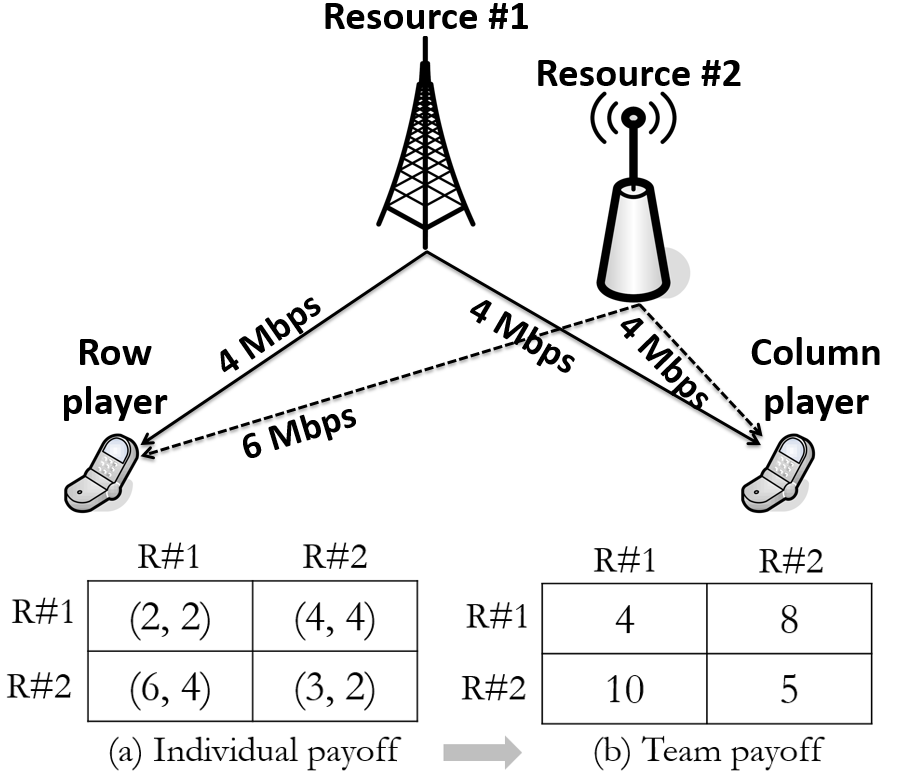}
 \caption{A simple $2\times2$ resource selection game.}
 \label{fig:2x2resoucegame}
\end{figure}

As each player separately optimizes its objective function based on the local view, the solution achieved might not be the most optimal one from the network's point of view. In this game, only one action combination is socially optimal, which is when the row player selects R\#2 and the column player selects R\#1, however there exist two possible NE outcomes: (R\#1, R\#2) and (R\#2, R\#1). With the later outcome, the row and column players respectively receive utilities of $6$ and $4$ Mbps, which yields a largest total network throughput of $10$ Mbps for the game. In this case, we cannot switch to any other outcome and make at least one party better off. %without making anyone worse off. 

Note that although (R\#2, R\#1) is the best solution for all players, given their local views, neither of them would unilaterally deviate from the suboptimal NE (R\#1, R\#2) since its individual utility would drop to a lower value. Because of that, players can often get stuck in an inferior equilibrium. In order to get out of that suboptimal point, they may need an outside intervention that provides a strong incentive for all to change their behavior. The incentive mechanism can just be a payoff-sharing scheme that rewards all participants with a same joint utility, which motivates all to act as part of a group. When considering joint team utilities, the utility matrix for a game changes, as shown on the right table in Fig.~\ref{fig:2x2resoucegame}. Thus, if the players could jointly agree on what to do, then it makes sense for them to choose a payoff-maximizing strategy at the (R\#2, R\#1) solution. 

%%%%%%%%%%%%%%%%%%%%%%%%%%%%%%%%%%%%%%%%%%%%%%%%%%%%%%%%%%%%%%%%%%%%%%%%

\section{Proposed Algorithm}

Our algorithm is a variation of external RM~\cite{hart2000simple} and relies on the following two assumptions: 

\textbf{\emph{Assumption 1:}} All players have access to their own utility functions. At every time $t$, every players $i$ is able to compute its own payoff, i.e., $u^i(a^i_t,a^{-i}_t)$, and the expected payoff it would have obtained if it had played any other action $k \neq a^i_t$, i.e., $u^i(k,a^{-i}_t), \forall k\in \Ac_i$, given that the other players do not change their decisions.

\textbf{\emph{Assumption 2:}} Every player shares extra information with other players to help them make better decisions and steer the collective behavior of all players to a desired state. Specifically, at every decision time step $t$, given the current action profile $(a_i^t,a_{-i}^t)$, every player $i$ computes a message consisting of its expected local payoff vector of all possible actions that a different player $j\neq i$ may take, defined as $\left\{\tilde{u}_i(a_i^t,a_{-i}^t|{a_j=k})\right\}_{k\in{\Ac_j}}$, and sends that message to the others (one message for each different player). 

We present a high-level structure of the learning technique to reach a desired social optimum solution in Algorithm~\ref{alg:algorithm}. The major differences between this proposed cooperative approach and a non-cooperative approach as in the standard RM are as follows:

(a) Firstly, we propose to use the global utility instead of the local utility to update players' regrets. This design makes sure the algorithm can find a stable PSNE since at least one PSNE outcome is guaranteed to exist. 

(b) Secondly, we add a small probability of a random action to ensure adequate exploration. Each player decides its next action with a high probability based on regret minimization, but keep exploring other options randomly with uniform probability, by maintaining a small exploration rate ($\frac{\delta}{t^\gamma} \frac{1}{|\Ac_i|}$), which decreases sufficiently slowly towards zero over time. This essentially allows players to keep searching for better outcomes, and avoid getting stuck in a local optima when all the regrets of choosing different actions vanish (i.e., converging to an arbitrary sub-optimal NE).   

(c) Thirdly, every player keeps track of its highest obtainable global utility over time and synchronizes this information across all players. This parameter is used as a global satisfaction threshold for all players so that whenever a player's current decision satisfies this threshold, it is forced to stay with this decision regardless of its regrets. This helps players to rapidly locate and verify whether a specific point in the joint action set is a PSNE outcome. Thus, it may speed up the convergence time of computing a PSNE.    

%--------------------------------------------------------------------------------
\begin{algorithm}[!h]
\DontPrintSemicolon
% \KwInitialise {$\pi^i_1(k)\leftarrow \frac{1}{|\Ac_i|},\ \forall k\in\Ac_i$}
\KwInput {The game $\Gc = \big(\Nc,\{\Ac_i\}_{i\in\Nc},\{u_i\}_{i\in\Nc}\big)$}
\KwOutput {Pareto-optimal strategy $(\pi^*_i,\pi^*_{-i}) \in \Ac$}
Initialise:\ {$U \leftarrow 0;\ \omega \leftarrow 0 $}\\ %\mbox{and}\ R_i^0(k) \leftarrow 0,
\While{$(\omega < 1)$} %\label{whileLoop} %\tcp*{this is a comment}
{
    $U \leftarrow U + \Delta U$\\
    Set:\ {$\begin{cases} \pi_i^1(k) \leftarrow 1/|\Ac_i| \\ R_i^1(k) \leftarrow 0 \end{cases}, \forall i\in\Ac_i\ \mbox{and}\ \forall k\in\Ac_i$}\\
    \For{$t \in \{1, 2, \dots, T\}$}
    {
      \For{$i\in\Nc$}
      {
          Sample and exchanges $a_i^t \sim \pi_i^t$ \\
          Compute local payoff $u_i(a_i^t,a_{-i}^t)$ \\
          Compute and send to all players $j\neq i$\\ $\lbr \tilde{u}_i(a_i^t,a_{-i}^t)\big|{a_j=k})\rbr_{k\in{\Ac_j}}$\\
          Compute expected global payoff $\forall k\in \Ac_i$\\ $u^g(k,a_{-i}^t) = \sum_{j\in\Ac} u_j(a_j^t,a_{-j}^t)\big|a_i=k $ \\
          \If{$\displaystyle\max_{i\in\Nc} \lbr \displaystyle\max_{k\in\Ac_i} \lbr u^g(k,a_{-i}^t) \rbr \rbr > U$}
          { $U \leftarrow \displaystyle\max_{i\in\Nc} \lbr \displaystyle\max_{k\in\Ac_i} \lbr u_i^g(k,a_{-i}^t) \rbr \rbr$}
          
          Compute cumulative regret $\forall k\in \Ac_i$\\ 
          $R_i^t(k) \leftarrow \displaystyle\frac{1}{t} \sum\nolimits_{\tau = 1}^t \Big[ u^g(k,a^{-i}_t) - u^g(a^i_t,a^{-i}_t)  \Big]$\\
          Update action selection strategy $\forall k\in\Ac_i$\\
          \uIf{$u_i^g(a^i_t,a^{-i}_t)\geq U$}
          { $\pi_i^{t+1}(k) \leftarrow \Ibb_{\lbr k=a_i^t \rbr}$ }          
          \uElseIf{$\sum_{k} \big[R^i_t(k)\big]^+=0$}
          { $\pi_i^{t+1}(k) \leftarrow 1/|\Ac_i|$ } %\displaystyle\frac{1}{|\Ac_i|}$ }
          \Else
          { $\pi_i^{t+1}(k) \leftarrow \lb 1-\displaystyle\frac{\delta}{t^\gamma} \rb \displaystyle\frac{\big[R^i_t(k)\big]^+}{\sum\nolimits_{k'} \big[R^i_t(k')\big]^+} + \frac{\delta}{t^\gamma} \frac{1}{|\Ac_i|}$ }
      }
      Update $\omega \leftarrow u^g(a_i^t,a_{-i}^t)/U$\\
      \If{$\omega < 1$} 
      { $(\pi_i^*,\pi_{-i}^*) \leftarrow (\pi_i^{t+1},\pi_{-i}^{t+1})$\\
      \textbf{go to} 2}
      \Else 
      { $(\pi_i^*,\pi_{-i}^*) \leftarrow (\pi_i^{t+1},\pi_{-i}^{t+1})$ }
    }
}
\Return{$(\pi_i^*,\pi_{-i}^*)$}\;
\caption{Social-Optimum Finding Algorithm} \label{alg:algorithm}
\end{algorithm}
% \Ibb_{\lbr k=a_i^t \rbr} + \Ibb_{\lbr k\neq a_i^t \rbr}
% { $\pi_i^{t+1}(k) \leftarrow \lb 1-\displaystyle\frac{\delta}{t^\gamma} \rb \Ibb_{\lbr k=a_i^t \rbr} + \displaystyle\frac{\delta}{t^\gamma}
%--------------------------------------------------------------------------------  
In more detail, our algorithm works in two main phases: (1) finding a feasible solution, and then (2) iteratively improving it until an optimal solution is found, as follows: 
% https://www.cs.cmu.edu/~./nilanjan/pubs/conference/noam_socialcom12.pdf

\textbf{\emph{Phase 1:}} Each player computes its vector of expected global utilities for all possible actions based on the information exchanged from other players. %The global utility is defined as the total sum of the individual utilities obtained by all players.
All players then update their regrets using the estimated global utilities instead of the local utilities. This means that a player can improve its own utility by unilateral action if and only if this unilateral action also improves the global utility. Setting all players' individual utilities to the global utility, with $u_i(a) = u^g(a)$, means that $\forall i\in\Nc, \forall a_i^{\prime}, a_i^{\prime\prime}\in\Ac_i,  \mbox{ and } \forall a_{-i}\in\Ac_{-i}$
\begin{multline*}
u_i(a_i^{\prime},a_{-i}) - u_i(a_i^{\prime\prime},a_{-i})>0\ 
\Leftrightarrow\ u^g(a_i^{\prime},a_{-i}) - u^g(a_i^{\prime\prime},a_{-i})>0\ .
\end{multline*}
This condition coincides with the definition of ordinal potential games~\cite{monderer1996potential}, with the potential function $\Phi(\cdot) \equiv u^g(\cdot)$, and hence the player utilities are aligned with the global objective. This leads to a desired result that the game would have at least one PSNE outcome as proved later in Lemma~\ref{Existence_PSNE}. This alignment also guarantees that the social optimum action profile $a^* \in \argmax_{a\in\Ac} u^g(a)$ (i.e., an action profile that maximises the global utility) are always included in the PSNE set, which is obvious to verify.  

\textbf{\emph{Phase 2:}} Each player monitors an experienced global utility using a synchronized parameter $U$. The main idea is letting each player $i$, at each time $t$, track and synchronize the highest possible global utility. Whenever $u_i^g(\cdot)\geq U$, the player $i$ will stay with its current action decision $a^i_t$ with probability of $1$, regardless of its regret. Once a pure action profile is reached, the monitored threshold $U$ is incremented by a step size $\Delta U$. Subsequently, the principal loop is performed, where an external RM algorithm is used to obtain a feasible convergence point. At this stage, a global satisfaction indicator $\omega$ is used to verify whether the social optimum is reached. This process is repeated until $\omega = 1$, which implies that the converged action profile $(a_i^*,a_{-i}^*)$ is the social-optimum PSNE. %, either a Nash or non-Nash point,

\subsection{Main Results}
%In the following, we show that our proposed algorithm can guarantee to reach an social optimal PSNE solution in a repeated coordination game. 
The following two lemmas are used to prove our main theorem. %\footnote{Due to space limits, we move the proofs to the appendix.} 

% \noindent\textbf{Existence of Pure Strategy Nash Equilibrium (PSNE)}
\begin{lemma}\label{Existence_PSNE}
A finite coordination game will always have at least one PSNE, if maximizing players' local utilities corresponds to maximizing the global objective, i.e., the players' local utility functions satisfy, $\forall a_i, a^{\prime}_i\in\Ac_i,\ \forall a_{-i}\in\Ac_{-i},\ \forall i\in \Nc\ ,$
\begin{multline} \label{eq:potential_function}
% u_i(a_i,a_{-i})-u_i(a^{\prime}_i,a_{-i})\\ 
% = w_i \Big( \Phi(a_i,a_{-i})-\Phi(a^{\prime}_i,a_{-i})\Big) ,
u_i(a_i,a_{-i})-u_i(a^{\prime}_i,a_{-i}) > 0\ 
\Rightarrow\ \Phi(a_i,a_{-i})-\Phi(a^{\prime}_i,a_{-i}) >0\ ,
\end{multline}
where $\Phi(\cdot)$ is a function that represents the global objective. %and $w=\{w\}_{i\in\Nc}$ is a positive vector of weights.
\end{lemma}
%\footnote{In other words, an improvement in utility of any player for its two different actions always results in an improvement in the global utility, given that other players' actions are unchanged.}

\begin{proof}
See Appendix A.1 in the supplemental material.
\end{proof}

\begin{lemma}\label{Regret_Convergence}
If a player $i$ updates its decisions according to Algorithm~\ref{alg:algorithm}, then its time-average regret with respect to every $a\in\Ac^i$ vanishes as $T\rightarrow \infty$. When all players use the same algorithm, the empirical distribution of the joint actions converges to a set of CCE.
\end{lemma}

\begin{proof}
    For simplicity of notation, we drop the subscript $i$ and define the Lyapunov function: \footnote{$\mbox{dist}(x,\Ac) = \min\{\|x-a\|:a\in\Ac\}$, where $\|\cdot\|$ is the Euclidean norm.}
	\begin{equation}
	\label{eq:Lyapunov}
	P(\bar{R}_t) = \frac{1}{2}\left(\mbox{dist}[\bar{R}_t,\Rbb^-]\right)^2 = \frac{1}{2}\sum\nolimits_{k}\left(|\bar{R}_t(k)|^+\right)^2 \ .
	\end{equation}
	Taking the time-derivative of~\eqref{eq:Lyapunov} yields
	\begin{equation}
	\label{eq:Pderivative}
	\frac{d}{dt}P(\bar{R})=\sum\nolimits_{k}|\bar{R}(k)|^+ \times \frac{d}{dt}\bar{R}(k) \ .
	\end{equation}
	First, we find $d\bar{R}(\theta)/dt$ as follows
	\begin{align*}
	& \frac{d}{dt}\bar{R}(k) = \Er_{\pi} \lbr R(k) - \bar{R}(k) \rbr \nonumber\\
	& = \Er_{\pi} \lbr \lsb u^g(k,a^{-i})-u^g(a) \rsb - \bar{R}(k) \rbr \nonumber\\
	& = \sum\nolimits_{a\in\Ac}\pi(a) \lsb u^g(k,a^{-i}) - u^g(a) \rsb - \bar{R}(k) \nonumber\\
	& = \sum\nolimits_{a^{-i}\in\Ac^{-i}} \pi^{-i}(a^{-i})\ u^g(k,a^{-i}) -\sum\nolimits_{a\in\Ac} \pi(a)\ u^g(a) - \bar{R}(k) \nonumber\\
	& =  \lsb u^g(k,\pi^{-i})-u^g(\pi) \rsb - \bar{R}(k) \ . %
	\end{align*}	
	Next, substituting $d\bar{R}(k)/dt$ into~\eqref{eq:Pderivative}, we obtain
	\begin{align}
	\frac{d}{dt}P(\bar{R}) = & \sum\nolimits_{k}|\bar{R}(k)|^+ \times \lsb u^g(k,\pi^{-i})-u^g(\pi) \rsb \nonumber\\ 
	& - \sum\nolimits_{k}|\bar{R}(k)|^+ \times \bar{R}(k) \ .
	\label{eq:Pderivative2}
	\end{align}
	
	Also, recall that the action selection strategy of the player is defined based on its regret function 
	$$\pi^i(k)=(1-\delta_n)\ \frac{|\bar{R}(k)|^+}{\sum_{k}|\bar{R}(k)|^+} + \frac{\delta_n}{m} \ .$$
	
	Thus, 
	$$ |\bar{R}(k)|^+ = \lsb \lb \frac{1}{1-\delta_n} \rb \pi^i(k) -\frac{\delta_n}{(1-\delta_n)m} \rsb \sum_{k}|\bar{R}(k)|^+ \ . $$
	
	Substitute $|\bar{R}(k)|^+$ into the first term on the r.h.s of~\eqref{eq:Pderivative2}:
	\begin{align}
	& \sum\nolimits_{k}|\bar{R}(k)|^+ \times \lsb u^g(k,\pi^{-i})-u^g(\pi) \rsb \nonumber \\
	& = \sum\nolimits_{k} \lsb \lb \frac{1}{1-\delta_n} \rb \pi^i(k) -\frac{\delta_n}{(1-\delta_n)m} \rsb \sum\nolimits_{k}|\bar{R}(k)|^+ \nonumber\\ 
	& \ \ \ \times \lsb u^g(k,\pi^{-i})-u^g(\pi) \rsb \nonumber \\
	& = \lb \frac{1}{1-\delta_n} \rb \sum\nolimits_{k}|\bar{R}(k)|^+ \sum\nolimits_k \pi^i(k) \lsb u^g(k,\pi^{-i})-u^g(\pi) \rsb \nonumber\\
	& \ \ \ - \frac{\delta_n}{(1-\delta_n)m}\ \sum\nolimits_{k}|\bar{R}(k)|^+ \lsb u^g(k,\pi^{-i})-u^g(\pi) \rsb \nonumber\\
	& \leq \frac{2M\delta_n}{(1-\delta_n)m}\ \sum\nolimits_{k}|\bar{R}(k)|^+ \ .
	\label{eq:firstterm}
	\end{align}
	Note that in the last line, the following results are used:
	\begin{align}
	& (1)\ \sum\nolimits_k \pi^i(k) \lsb u^g(k,\pi^{-i})-u^g(\pi) \rsb=0 \ \mbox{\big(linearity of $u^g(\cdot)$\big)}\ . \nonumber\\
	& (2)\ \lsb u^g(k,\pi^{-i})-u^g(\pi) \rsb \leq 2M \ \mbox{\big($u^g(\cdot)$ is bounded by $M$\big)}\ . \nonumber
	\end{align}
	
	Now consider the last term on the r.h.s of~\eqref{eq:Pderivative2}
	\begin{align}
	\sum\nolimits_{k}|\bar{R}(k)|^+ \times \bar{R}(k)=\sum\nolimits_{k}\left(|\bar{R}(k)|^+\right)^2 \ . %=2P(\bar{R})
	\label{eq:secondterm}
	\end{align}
	
	Therefore, combining~\eqref{eq:Pderivative2},~\eqref{eq:firstterm} and~\eqref{eq:secondterm}, we obtain
	$$\frac{d}{dt}P(\bar{R}) \leq \sum\nolimits_{k}|\bar{R}(k)|^+ \lb \frac{2M\delta_n}{(1-\delta_n)m}\ - \sum\nolimits_{k}|\bar{R}(k)|^+ \rb \ .$$ 
	
	Assuming $\sum\nolimits_{k}|\bar{R}(k)|^+ \geq \epsilon > 0$, one can choose $\delta_n > 0$ small enough such that:
	$$\frac{d}{dt}P(\bar{R}) \leq -\frac{1}{2} \sum\nolimits_{k} \lb |\bar{R}(k)|^+ \rb^2 = P(\bar{R}) \ ,$$ 
	by substituting $\sum\nolimits_{k}\big(|\bar{R}(k)|^+\big)^2 = 2 P(\bar{R})$ from~\eqref{eq:Lyapunov}.
	
	Consequently, 
	$P\big(\bar{R}(t)\big)= P(\bar{R}(0))\ \exp(-t) \ .$
	This implies that $P\big(\bar{R}(t)\big)$ approaches zero at an exponential rate. This proves approachability of the player' regrets to the set $\Rbb^-$ (i.e. all regrets approach zero).
	
	It follows that, if all players use the same proposed procedure, we obtain the convergence of the empirical distribution of the joint actions of all players to approach the {\it Coarse Correlated Equilibrium set} of correlated actions yielding non-positive rewards. The result is immediate from the definition of CCE as in~\eqref{eq:CCE}. 
    On any convergent subsequence $\displaystyle\lim_{t \rightarrow \infty} \pi_t \rightarrow \Pi$,
    $$\displaystyle\lim_{t \rightarrow \infty} \Big[R(k,\pi_{-i}^t)\Big] = \sum\nolimits_{a \in \Ac} \Pi(a^t) \lb u^g(k,a^t_{-i}) - u^g(a^t) \rb \leq 0 \ , $$
    where  $R(k,\pi_{-i}^t)$ is the $k$-th component of the player's time-average regret vector. \footnote{The regret that a player would have by playing action $k$ with probability one, instead of the (generally randomised) action specified by it strategy.} Comparing with the definition of the CCE as defined in~\eqref{eq:CCE}, the desired result follows.
\end{proof}

%http://www2.hawaii.edu/~gurdal/JDSMC07.pdf

We now state and prove our main result.
\begin{theorem}\label{thrm_socialOptimum}
If all players update their decisions according to Algorithm~\ref{alg:algorithm}, then the resulting action profile is guaranteed to converge to a social optimum PSNE solution.
\end{theorem}

\begin{proof}
% We divide our approach into two particular cases:
% (1) If all players have access to the global utility function (the total sum of the players' local utilities)
Converging to a social optimum PSNE in our method requires a combination of two conditions: (1) the algorithm converges to a PSNE, and (2) a value of $\omega = 1$ is reached. The later condition necessarily guarantees that the players have played to reach an upper bound $M$ of the game objective function, which is the maximum of the global utility function $u^g(\cdot)$. We now proceed to prove whenever $u^g(a_i^*,a_{-i}^*) = M$ is reached, using the proposed algorithm will lead to the convergence to a social-optimum PSNE. %result in all the players' unconditional regrets vanish, which

% The proof of Lemma~\ref{Regret_Convergence} (Appendix B) has shown that all the time-average regrets of a player implementing Algorithm~\ref{alg:algorithm} will vanish, regardless of the behaviour of other players. 

The proof is constructed in the following two steps:

\textbf{Step 1 (Convergence toward a PSNE):} We first prove that if all players play according to Algorithm~\ref{alg:algorithm}, then the time-average strategy converges to a PSNE.  %The proof is shown in Appendix~\ref{Regret_Convergence}. 

The proof of Lemma~\ref{Regret_Convergence} (see Appendix B) has shown that if all players use the same proposed procedure in Algorithm~\ref{alg:algorithm}, the empirical distribution of the joint actions of all players approaches the {\it CCE} set of correlated actions. %yielding non-positive rewards. 
%all the time-average regrets of a player running Algorithm~\ref{alg:algorithm} will vanish, regardless of the behaviour of other players. It follows that, if all players use the same proposed procedure, we obtain the convergence of the empirical distribution of the joint actions of all players to approach the {\it CCE} set of correlated actions yielding non-positive rewards. 
Now assuming the considered coordination game $\Gc$ is a socially concave game~\cite{hart2015markets}, which satisfies the following two socially concave properties:\footnote{Note that many nature games satisfy these conditions, such as linear resource allocation or network congestion control games.} 
\begin{itemize}
    \item \emph{Assumption 1}: $\sum_{i\in\Nc} \lambda_i u_i(a)$ is concave in $a$,
    \item \emph{Assumption 2}: $u_i(a_i,a_{-i})$ is convex in $a_{-i}$.
\end{itemize}
The combination of the two assumptions implies that player $i$'s local utility function is concave in $a_i$ given $a_{-i}$ is fixed. 

Let $a$ be a CCE of $\Gc$, and let $\bar{a} = \Er_{\pi} \lsb a\rsb$, we then prove that $\bar{a}$ is a pure NE of $G$. Without loss of generality, assume that $\lambda_i = 1,\ \forall i\in\Nc$. As $a$ is a CCE point, it satisfies
\begin{equation} \label{eq:CCE_property}
\Er \lsb u_i(a) \rsb \geq \Er \lsb u_i(a_i^{\prime},a_{-i}) \rsb \ ,
\end{equation}
for every $i\in\Nc$ and every action $a_i^{\prime}\in\Ac_i$. Also, since $\bar{a}\in \Ac$, using Assumption 2 we have
\begin{equation} \label{eq:CCE_property2}
\Er \lsb u_i(a_i^{\prime}, a_{-i} \rsb \geq u_i \big( a_i^{\prime}, \Er \lsb a_{-i} \rsb \big) = u_i(a_i^{\prime}, \bar{a}_{-i})\ .
\end{equation}

Combining~\eqref{eq:CCE_property} and~\eqref{eq:CCE_property2} yields
\begin{equation} \label{eq:CCE_property3}
\Er \lsb u_i(a) \rsb \geq u_i(a_i^{\prime}, \bar{a}_{-i}) \ .
\end{equation}

Replacing $a_i^{\prime} = \bar{a}_i$ and then summing over all $i\in \Nc$
$$\sum\nolimits_{i\in\Nc} \Er \lsb u_i(a) \rsb \geq \sum\nolimits_{i\in\Nc} u_i(\bar{a}_i, \bar{a}_{-i}) = \sum\nolimits_{i\in\Nc} u_i(\bar{a}) \ .$$

Using Assumption 1 implies 
$$\sum\nolimits_{i\in\Nc} \Er \Big[ u_i (a) \Big] = \Er \lsb \sum\nolimits_{i\in\Nc} u_i(a) \rsb \leq \sum\nolimits_{i\in\Nc} u_i \Big( \Er \lsb a \rsb \Big)\ . $$

Therefore
$$\sum\nolimits_{i\in\Nc} \Er \Big[ u_i (a) \Big] = \sum\nolimits_{i\in\Nc}u_i(\bar{a}) \ .$$

Thus, $u_i(\bar{a})=\Er\lsb u_i(a)\rsb$ for every $i$, and~\eqref{eq:CCE_property3} becomes
$$u_i(\bar{a})\geq u_i(a_i^{\prime}, \bar{a}_{-i})\ .$$
for every $a^{\prime}_i\in\Ac_i$. Therefore, $\bar{a}$ is a pure Nash equilibrium of $\Gc$. This implies that the time-average of the empirical distributions converges to a pure Nash equilibrium.

\textbf{Step 2 (Establishing convergence to a social optimum):} As the players do not know the upper bound $M$ of their global utility there are three particular cases:

$\underline{\mbox{Case 1}\ (U < M)}:$
When $U< M$ and the computed point $a$ is a non-Nash point, then $\exists a^{\prime}\in\Ac$ such that $u^g(a^{\prime})>u^g(a)$ and thus the players' regrets are positive. Therefore, the algorithm will continue to run until all players' regrets vanish, which results in the convergence to a PSNE point. 

$\underline{\mbox{Case 2}\ (U = M)}:$
When $U = M$, the computed point $a$ is a PSNE point, but if $\omega <1$ then the algorithm will not stop but repeat (line 27) until the value of $\omega = 1$ is reached. The proposed algorithm incorporates some randomness (lines 18--23) as part of its procedure, where players update their strategies according to their regrets with probability $(1-\delta_n)$ and randomly select an action from its action set otherwise. In case either all regrets vanish or the current decision meets the synchronized global utility threshold $(u^g(a_i^t,a_{-i}^t)\geq U)$, the player retains its current decision. Thus, if the learning procedure is repeated a sufficient number of times, then the subsequent converged point with a higher global utility will replace the previous computed point as a better PSNE outcome. This procedure is repeated until it reaches the social-optimum PSNE resulting in $\omega =1$.  

$\underline{\mbox{Case 3}\ (U > M)}:$
When $U > M$, the value of $\omega = 1$ cannot be found and thus the algorithm will stop at a random PSNE point. In this case, the best action profile found so far with $\omega =1$ is the social-optimum PSNE outcome. 
\end{proof}

%===================================
\subsection{Communication Complexity Trade-off}
%We discuss here the communication complexity and algorithm run-time of our proposed method. 
A multiplayer normal-form game can be respresented with a tensor $\Uc$ in which each entry $u_i(a_i, a_{-i})$ specifies an individual payoff for player $i$ under the action profile $(a_i, a_{-i})$. The number of action profiles grows exponentially in the number of players. A game with $n$ players and at most $m$ actions per player has a total of $|\Ac|=m^n$ action profiles. Each player has one individual payoff for each action profile and thus it requires $n\times m^n$ integer numbers to represent all players' payoffs. Therefore, the complexity of searching over the space of all possible action profiles is $\mathcal{O}(n\ m^n)$, which is generally intractable both for storing a game and computing its solution. 

Our algorithm follows the regret based iterative learning rule, which has shown to be able to compute an approximate CCE using only $poly(n, m)$ (polynomial in $n$ and $m$) payoff queries~\cite{hart2010long, babichenko2016query}. For steering all players' collective behavior to a desired equilibrium, i.e, the CCE with maximum social welfare, extra information is required to exchange among the players. In order to demonstrate the communication complexity trade-off of our approach, we evaluate the ratio of the number of entries queried by all players to the number of entries in $\Uc$ at each decision point. Using our method, at each iteration every player $i$ has to send counterfactual payoffs for every alternative action to all other ($n-1$) players, which incurs the total number of $n(n-1)m$ queries. Thus, the resulting ratio is $\displaystyle\frac{n-1}{m^{n-1}}\ll 1$ (for a large game scenario). The total amount of queries over time (as a function of the number of players $n$, the number of actions per player $m$, and the number of iterations until convergence $T$) is therefore less than~$\sim\mathcal{O}(m\ n^2\ T)$. As a comparison reference, an equivalent number of queries over time required by the standard Regret Matching algorithm as in~\cite{hart2000simple} is~$\sim\mathcal{O}(n^2\ T)$. %Overall, the communication complexity of our proposed algorithm grows linearly in the square of the number of players and in the number of actions per player.   

% https://www.ifaamas.org/Proceedings/aamas2022/pdfs/p507.pdf
% https://arxiv.org/pdf/1306.6686.pdf

%The state-of-the-art literature demonstrates that arbitrary CCE are computable via iterative learning rules, such as RM~\cite{hart2000simple}. However, there is no performance guarantee for finding the global-optimal solution, i.e., the CCE with maximum social welfare (sum of all players' payoffs), among the set of CCE. \HN{ You need to go into more technical detail than this. Why can't it provide any guarantee? what you hope to change?}

%Third, another challenge of cooperative control in the game-theoretic approach lies in designing both the local objective functions for each of the individual players and their learning dynamics so that they collectively accomplish a desired global objective. In a repeated game, the amount of information that players acquire in repeated plays directly influences the resulting learning dynamics.

%%%%%%%%%%%%%%%%%%%%%%%%%%%%%%%%%%%%%%%%%%%%%%%%%%%%%%%%%%%%%%%%%%%%%%%%

\section{Experimental Results}
We numerically test our method on two different applications: (1) Resource allocation game, where the number of player is much larger than the number of limited resources, and (2) Combinatorial task assignment game, where the number of players is much smaller than the number of available tasks. 

%===================================
\subsection{Multi-Agent Resource Allocation Game}
We use a large-scale network selection game (illustrated in Fig.~\ref{fig:simple_graph}) in the application of resource optimization in wireless networks. Here, there are $n$ mobile users (players) seeking to access network services from $m$ base stations (actions). We consider user throughput as player's utility and adopt the proportional-fair throughput model as in~\cite{nguyen2017reinforcement,7878785,8489869}. Under this utility model, each player obtains a different user-specific utility which is defined as user's maximum physical rate (generated randomly in the range $[1, 100]$ Mbps) divided by the total number of players selecting the same resource. %, to numerically test our proposed method.
\begin{figure}[!h]
 \centering
 \includegraphics[width=0.6\linewidth]{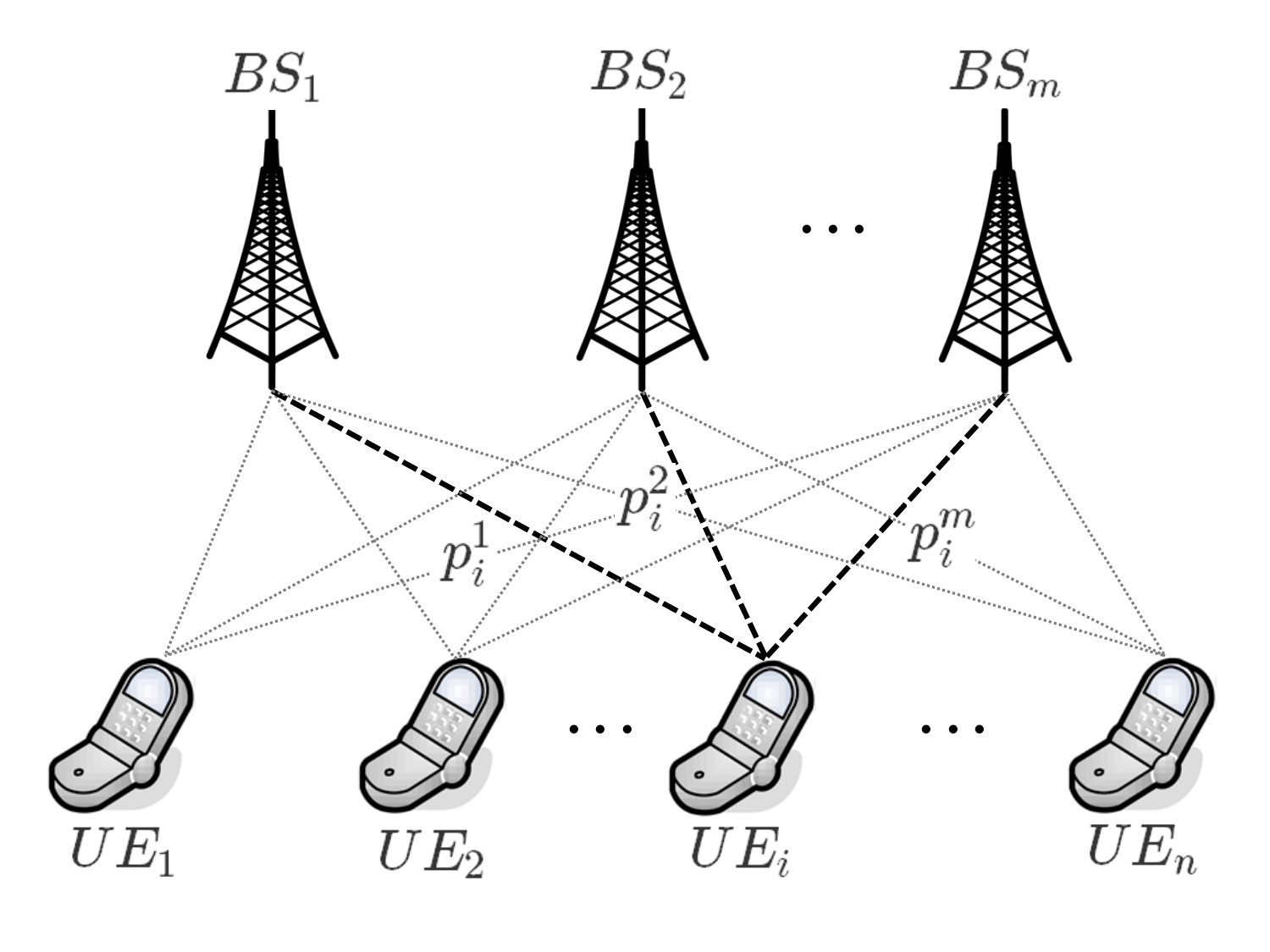}
 \caption{A $n\times m$ resource allocation game in wireless networks consisting of $n$ end-users and $m$ base stations.}
 \label{fig:simple_graph}
\end{figure}

We first demonstrate the performance of our method in achieving a social optimum outcome on a small $2\times 2$ resource allocation game as illustrated in the example in Fig.~\ref{fig:2x2resoucegame}. Fig.~\ref{fig:ResourceAllocation2players} demonstrated the convergence of our method toward the social optimum outcome $(R\#2, R\#1)$, which provides the maximum network throughput of $10$ Mbps, instead of $8$ Mbps as offered by the other suboptimal NE $(R\#1, R\#2)$. Meanwhile, the result of running the standard RM randomly converges to one of the two points.   
\begin{figure}[!h]
  \centering
%   \begin{subfigure}[b]{0.493\linewidth}
%     \includegraphics[width=0.9\linewidth]{figures/resource-allocation2.PNG}
%     % \caption{More coffee.}
%   \end{subfigure} 
%   \begin{subfigure}[b]{0.493\linewidth}
%     \includegraphics[width=\linewidth]{figures/ResourceAllocation_choice}
%     % \caption{More coffee.}
%   \end{subfigure}   
  \begin{subfigure}[b]{0.493\linewidth}
    \includegraphics[width=\linewidth]{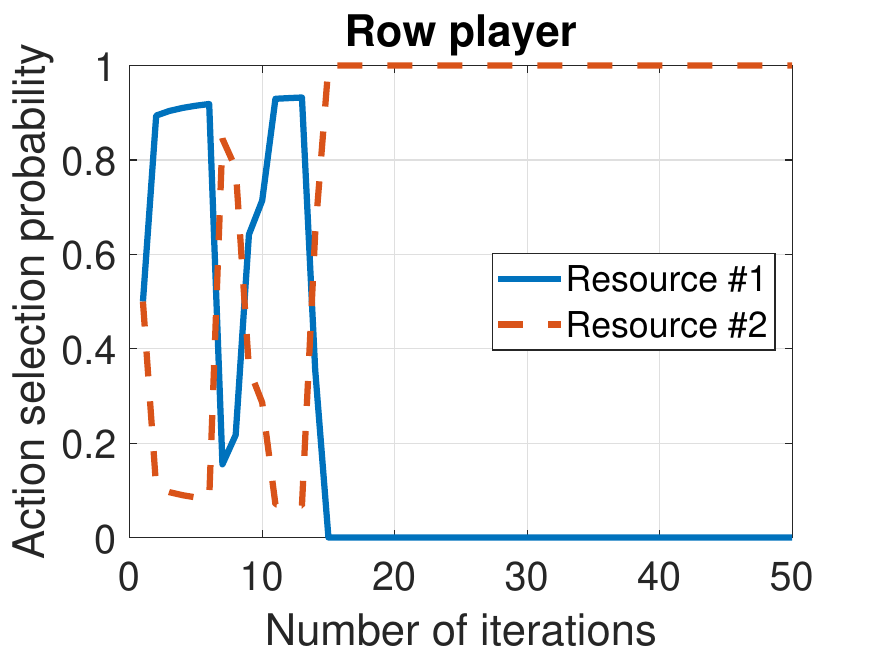}
    % \caption{Coffee.}
  \end{subfigure}
  \begin{subfigure}[b]{0.493\linewidth}
    \includegraphics[width=\linewidth]{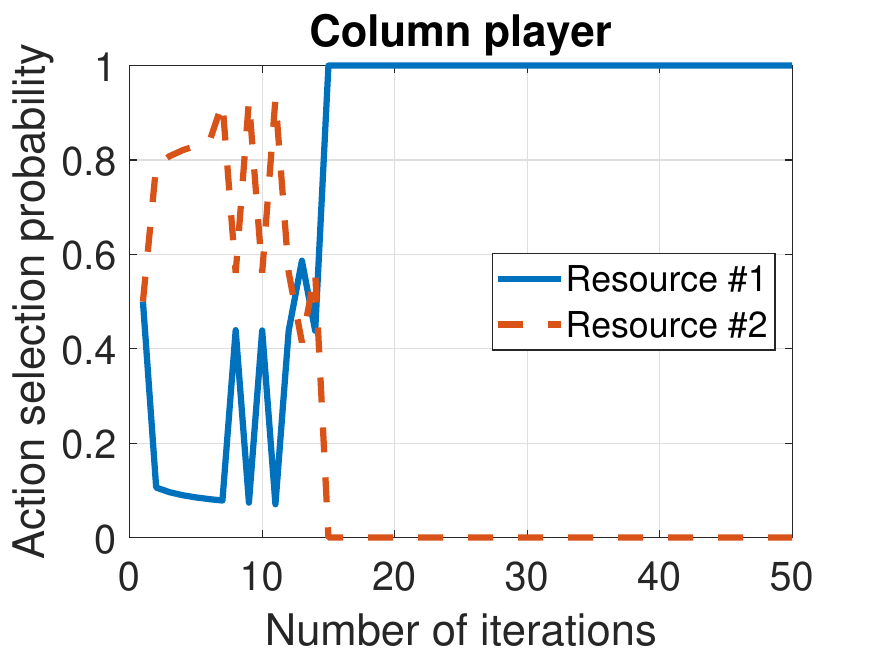}
    % \caption{More coffee.}
  \end{subfigure}
  \caption{Evolution of action selection probabilities of all players versus iterations in the $2\times2$ resource selection game as illustrated in the example in Fig.~\ref{fig:2x2resoucegame}.}
  \label{fig:ResourceAllocation2players}
\end{figure}

 We now test the performance of our proposed algorithm on the same game at a large scale. We compare the performance of the following regret minimization methods:
\begin{itemize}
    \item \textbf{SORM}: This is our proposed algorithm, which makes three major modifications to the standard RM for computing a social-optimum equilibrium outcome.
    \item \textbf{LURM}: This is the standard Regret Matching~\cite{hart2000simple}, in which players use local utilities to update their action strategies. 
    \item \textbf{GURM}: GURM makes one change to the standard RM, which uses the global utility function to update players' strategies. We included the comparison with the GURM scheme here to demonstrate that this modification alone is not sufficient to compute the desired outcome as provided in our method.
    %\item \textbf{Greedy RM}~\cite{zhang2022equilibrium}: Greedy RM is identical to standard RM except that each new iteration is greedily weighted to minimize the distance to the set of equilibria. Greedy RM has been shown to achieve faster convergence speed and higher social welfare than other existing regret-based approaches.
\end{itemize}
\begin{figure}[t!]
  \centering
  \begin{subfigure}[b]{0.493\linewidth}
    \includegraphics[width=\linewidth]{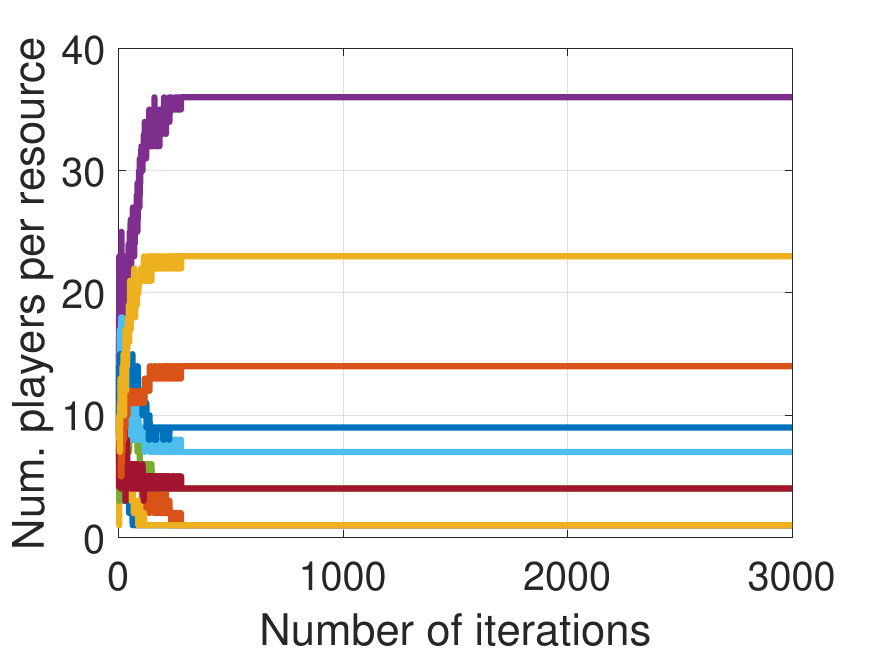}
    \caption{SORM}
  \end{subfigure} 
  \begin{subfigure}[b]{0.493\linewidth}
    \includegraphics[width=\linewidth]{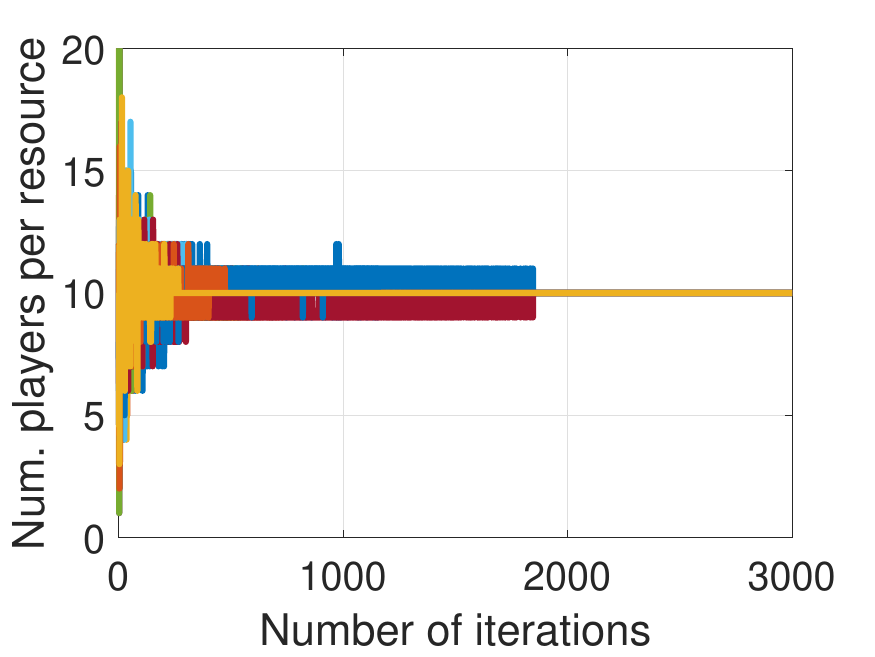}
    \caption{LURM}
  \end{subfigure}   
  \begin{subfigure}[b]{0.493\linewidth}
    \includegraphics[width=\linewidth]{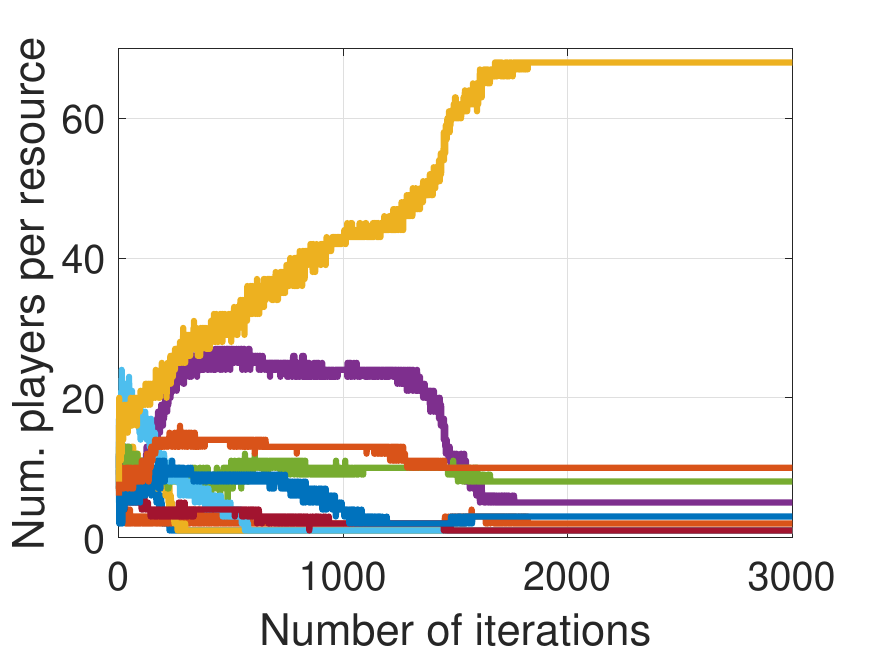}
    \caption{GURM}
  \end{subfigure} 
  \begin{subfigure}[b]{0.493\linewidth}
    \includegraphics[width=\linewidth]{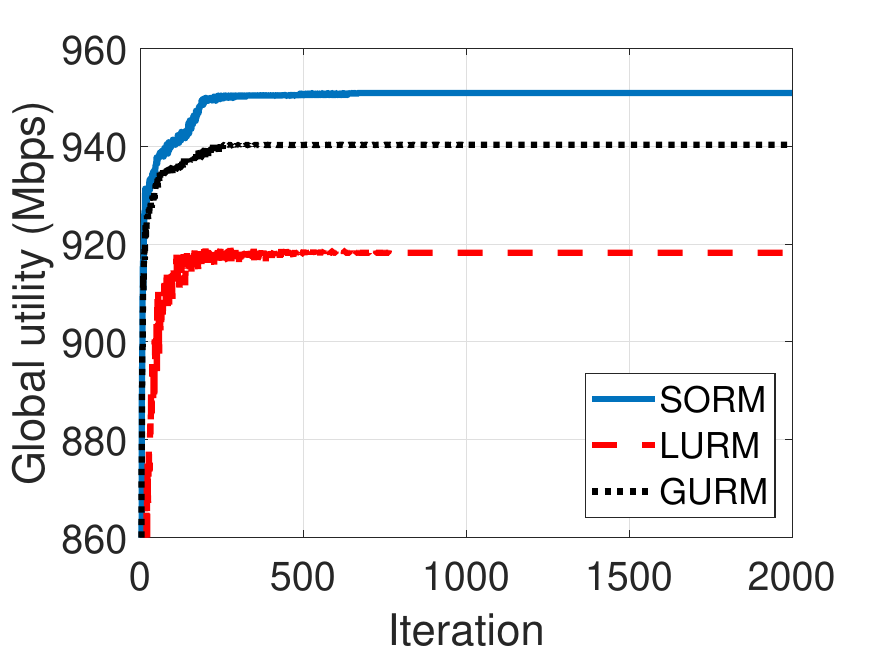}
    \caption{Global utility}
  \end{subfigure}   
  \caption{Illustration of convergence of different algorithms: (a) SORM, (b) LURM, (c) GURM; and (d) their achieved global utilities in a $100\times10$ resource allocation game.}
  \label{fig:ResourceAllocation100x10game}
\end{figure}

\begin{figure}[t!]
  \centering
  \begin{subfigure}[b]{0.493\linewidth}
    \includegraphics[width=\linewidth]{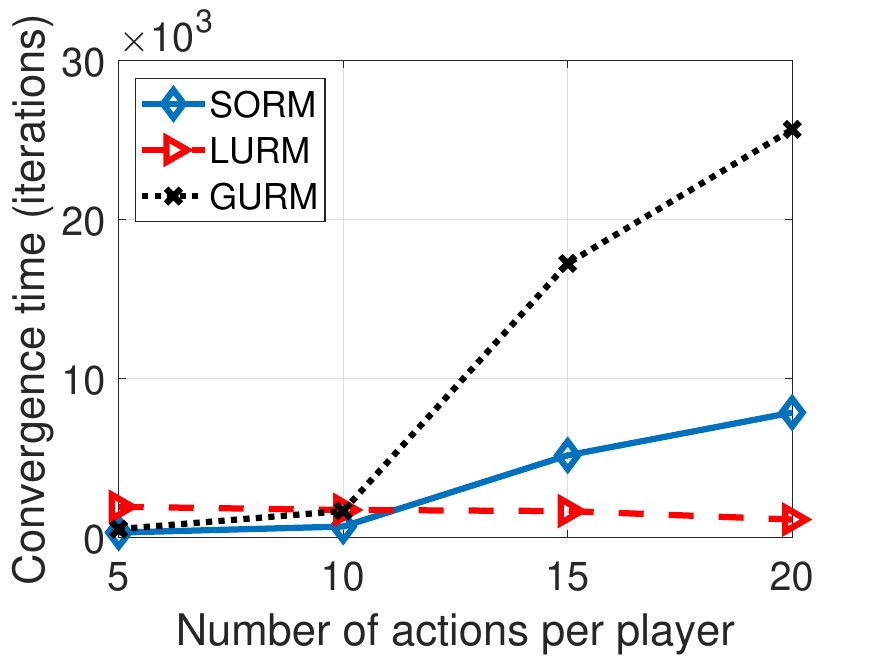}
    % \caption{Convergence speed}
  \end{subfigure} 
  \begin{subfigure}[b]{0.493\linewidth}
    \includegraphics[width=\linewidth]{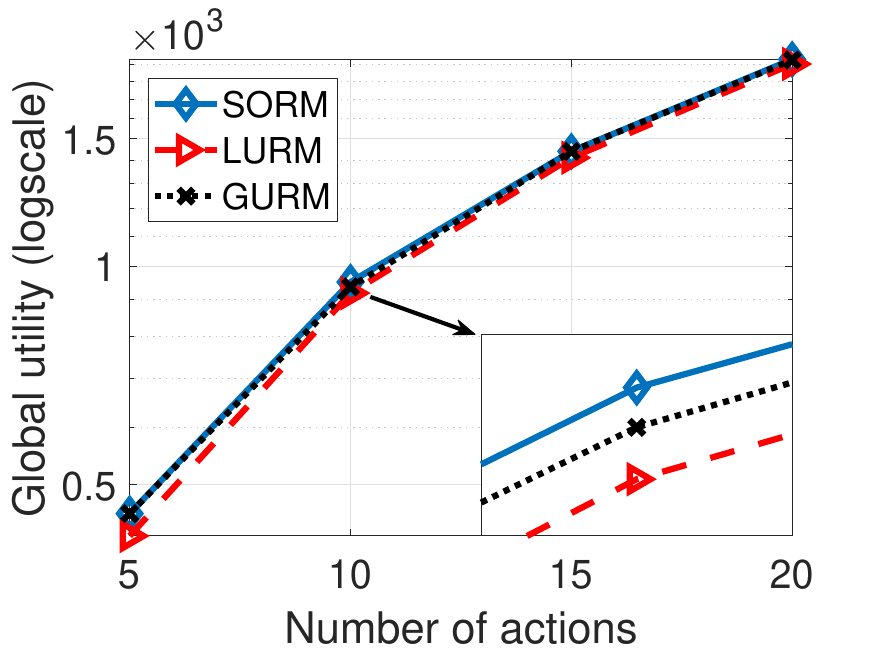}
    % \caption{All players' utilities}
  \end{subfigure}   
  \caption{Comparison of convergence time and achieved global utility when increasing the number of resources in a $100-$players resource allocation game.}
  \label{fig:ResourceAllocation_scalability}
\end{figure}

 %We follow the same proportional-fair throughput model in~\cite{nguyen2017reinforcement}, in which each player obtains a different user-specific throughput depending on its own physical rate  on the targeted resource and the number of players sharing the same resource. 

We conduct one random run of the above algorithms for a $100\times10$ resource allocation game and plot the convergence behavior (with regard to the number of players per resource) and the achieved global utility (sum of all player's utilities) versus iteration in Fig.~\ref{fig:ResourceAllocation100x10game}. As shown, our SORM obtains the highest performance both in convergence speed and global utility compared to the others. Both LURM and GURM algorithms converge at a similar speed. GURM, however, reaches the PSNE point that yields a higher global utility.      

To test how the different algorithms scale with regards to the number of action choices, we fix the player number at $100$ players and vary the number of actions per player from $5$ to $20$ resources. Figs.~\ref{fig:ResourceAllocation_scalability}(a) and~\ref{fig:ResourceAllocation_scalability}(b) respectively show the performances in convergence speed and global utility of the three algorithms. We first observe that, although the convergence speed of LURM is not much affected by increasing the number of available actions in the game, LURM always converges to a suboptimal social welfare outcome in all cases. Meanwhile, the remaining two algorithms can successfully find PSNE outcomes with much higher social welfare, with SORM achieving superior performance than GURM.   

In comparison of convergence rate with LURM, both GURM and SORM converge relatively fast when the number of action choices are small (i.e., 5 or 10 actions), however require a longer time to converge for a large number of actions per player. This is because the number of joint action profiles grows exponentially when the number of actions increases, and thus the number of potential PSNE points will also increase significantly. Consequently, the time required for searching for a good outcome over the space of potential PSNE profiles will increase. To deal with this long convergence issue for large games, we filter out all the noise (small increment in global utility) by setting all the small probabilities (i.e., $\leq 0.03$) of switching to other actions to zero, and normalize the vector of action probabilities accordingly. With this modification, as demonstrated in Fig.~\ref{fig:ResourceAllocation_scalability}, the convergence speed of SORM is almost linearly proportional to the number of actions per player, while the running time required for the GURM algorithm to converge grows significantly with the number of actions per player.

%===================================
\subsection{Multi-Agent Task Assignment Game}
We also conduct some experiments on games in which the number of players is much smaller than the number of action choices. We consider a hard combinatorial task assignment game where a group of $n$ agents needs to select tasks from a large set of available tasks, i.e., $m$ tasks with $m >> n$. Here, the global utility of an assignment profile is measured by the sum of individual task utilities and the objective is to find an assignment profile that maximizes the global utility. This problem is known as NP-hard in general, whose optimal solution is intractable to compute. 

We follow the same task utility model as described in~\cite{qu2019distributed}. We generate the set of agents and the set of targets uniformly random as point mass objects in a three dimentional box $[0, 1]^3$. Each target $k$ has a different value $V(k)$, which is uniformly sampled from $[10, 100]$. The utility of a target $k$ is the probability that it gets destroyed times its value
$$u(k)= V(k)\lb 1-\prod\nolimits_{i=1}^n\ p_{ik} \rb \Ibb_{\lbr k=a_i \rbr} \ .$$
Here, $p_{ik}=1/\big[ 1+exp\lb-\alpha\ \mbox{dist}(i,k)+\beta \rb \big]$ is the probability of survival of the target $k$ with respect to the agent $i$, which is an increasing function of the distance between $i$ and $k$. The parameter $\alpha, \beta$ characterizes how the survival probability increases with respect to the distance between the agent and the target. The global utility is thus can be defined as the total expected value of all the targets given the assignment profile:
$$u^g(a_i,a_{-i})= \sum\nolimits_{k=1}^{m}\ V(k)\lb 1-\prod\nolimits_{i=1}^n\ p_{ik} \rb \Ibb_{\lbr k=a_i \rbr} \ .$$
%to numerically test our proposed SORM algorithm.

We set $\alpha = 2$, $\beta = 1$, $V(k)\sim [10, 100]$ (uniformly random) and compare the performance of SORM versus GURM under two different cases:
\begin{itemize}
    \item Case I: we choose $n=20$ and $m =40$.
    \item Case II: we choose $n=40$ and $m=80$.
    % \item Case III: we choose $n=500$ and $m=1000$.
\end{itemize}
\begin{figure}[!h]
  \centering
  \begin{subfigure}[b]{0.493\linewidth}
    \includegraphics[width=\linewidth]{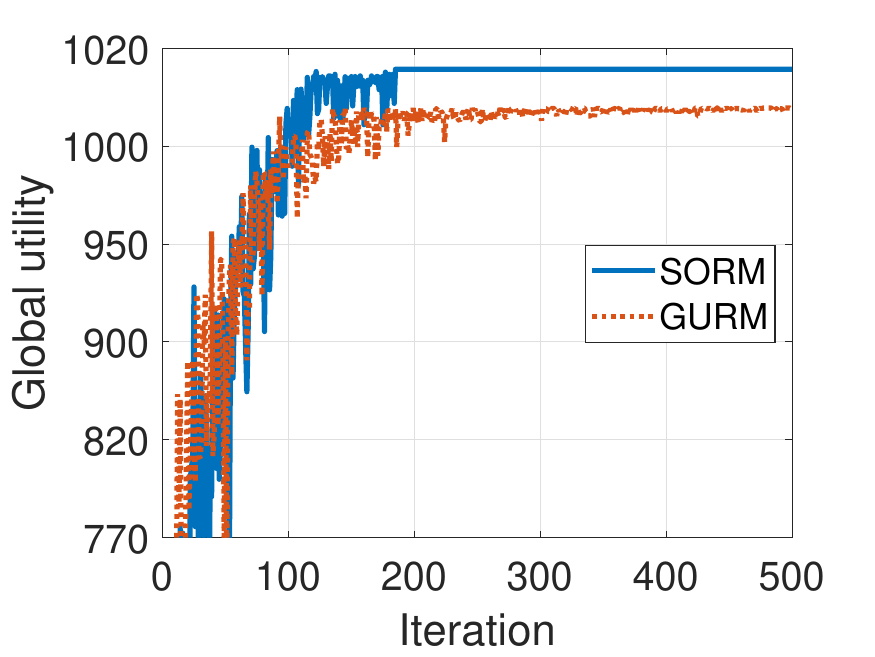}
    \caption{Case I: $n=20$ and $m =40$}
  \end{subfigure}
  \begin{subfigure}[b]{0.493\linewidth}
    \includegraphics[width=\linewidth]{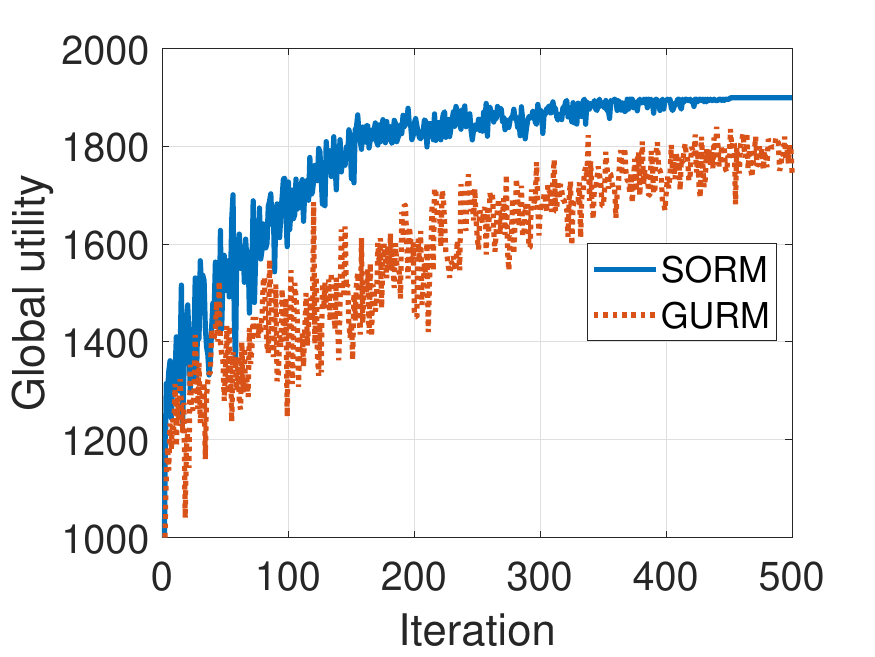}
    \caption{Case II: $n=40$ and $m =80$}
  \end{subfigure}
  \caption{Comparison on global utility between SORM and GURM in the multi-agent task assignment games.}
  \label{fig:WTAgame}
\end{figure}
% The simulation results are shown in Fig.~\ref{fig:WTAgame}. 
Note that we do not include the results of running LURM here as this algorithm was designed to solve non-cooperative games and thus results in very poor performance under this game setting as observed in our simulation. This can be explained by the fact that every player running LURM selfishly optimizes the utility of its own selected target, without considering the behaviors of the other players. Consequently, all players tend to choose a same target with highest value, which results in a very low value of global utility.  

It can be seen from Fig.~\ref{fig:WTAgame} that our proposed SORM algorithm has a clearly faster convergence speed, and achieves a better value of global utility in both the considered cases, as compared to that of the GURM method. Also, with the same running time of $500$ iterations, our SORM has successfully converged to a pure-strategy Nash action profile, while the GURM algorithm converges relatively slow and has not yet converged to a pure Nash point. Note that the GURM algorithm is not the standard Regret Matching approach (which is known to perform poorly in this cooperative game setting as discussed earlier), but is a modified RM algorithm which follows one (out of three) modifications proposed by our approach. These additional results clearly demonstrate the superior performance of our proposed method in various game settings.

%%%%%%%%%%%%%%%%%%%%%%%%%%%%%%%%%%%%%%%%%%%%%%%%%%%%%%%%%%%%%%%%%%%%%%%%

\section{Conclusions}
We introduced a new regret matching based method to compute a social optimum equilibrium for large-scale distributed multiagent optimization problems. Our results are general and can apply to many applications of coordination games. We theoretically proved that our method can guarantee convergence to a stable social optimal pure strategy Nash equilibrium (PSNE) whereas existing regret matching based methods typically converge to an arbitrary point or cycle among the set of attainable outcomes. We numerically show via two different large-scale resource optimization applications that our method can reach the PSNE outcome yielding highest social welfare compared to the other related methods. %We also find that our proposed method can obtain a faster convergence rate than the standard Regret Matching.  
For future research, we will also consider the multi-agent distributed control problem in a dynamic environment, in which the set of players' actions and their utilities evolve with time.

\section*{Acknowledgement}
Hung Nguyen is partially supported by DSTG grant ID10380.
%%%%%%%%%%%%%%%%%%%%%%%%%%%%%%%%%%%%%%%%%%%%%%%%%%%%%%%%%%%%%%%%%%%%%%%%

%%% The acknowledgments section is defined using the "acks" environment
%%% (rather than an unnumbered section). The use of this environment 
%%% ensures the proper identification of the section in the article 
%%% metadata as well as the consistent spelling of the heading.

% \begin{acks}
% If you wish to include any acknowledgments in your paper (e.g., to 
% people or funding agencies), please do so using the `\texttt{acks}' 
% environment. Note that the text of your acknowledgments will be omitted
% if you compile your document with the `\texttt{anonymous}' option.
% \end{acks}

%%%%%%%%%%%%%%%%%%%%%%%%%%%%%%%%%%%%%%%%%%%%%%%%%%%%%%%%%%%%%%%%%%%%%%%%

%%% The next two lines define, first, the bibliography style to be 
%%% applied, and, second, the bibliography file to be used.

\bibliographystyle{ACM-Reference-Format} 
\bibliography{reference}

%%% -*-BibTeX-*-
%%% Do NOT edit. File created by BibTeX with style
%%% ACM-Reference-Format-Journals [18-Jan-2012].

\begin{thebibliography}{30}

%%% ====================================================================
%%% NOTE TO THE USER: you can override these defaults by providing
%%% customized versions of any of these macros before the \bibliography
%%% command.  Each of them MUST provide its own final punctuation,
%%% except for \shownote{}, \showDOI{}, and \showURL{}.  The latter two
%%% do not use final punctuation, in order to avoid confusing it with
%%% the Web address.
%%%
%%% To suppress output of a particular field, define its macro to expand
%%% to an empty string, or better, \unskip, like this:
%%%
%%% \newcommand{\showDOI}[1]{\unskip}   % LaTeX syntax
%%%
%%% \def \showDOI #1{\unskip}           % plain TeX syntax
%%%
%%% ====================================================================

\ifx \showCODEN    \undefined \def \showCODEN     #1{\unskip}     \fi
\ifx \showDOI      \undefined \def \showDOI       #1{#1}\fi
\ifx \showISBNx    \undefined \def \showISBNx     #1{\unskip}     \fi
\ifx \showISBNxiii \undefined \def \showISBNxiii  #1{\unskip}     \fi
\ifx \showISSN     \undefined \def \showISSN      #1{\unskip}     \fi
\ifx \showLCCN     \undefined \def \showLCCN      #1{\unskip}     \fi
\ifx \shownote     \undefined \def \shownote      #1{#1}          \fi
\ifx \showarticletitle \undefined \def \showarticletitle #1{#1}   \fi
\ifx \showURL      \undefined \def \showURL       {\relax}        \fi
% The following commands are used for tagged output and should be
% invisible to TeX
\providecommand\bibfield[2]{#2}
\providecommand\bibinfo[2]{#2}
\providecommand\natexlab[1]{#1}
\providecommand\showeprint[2][]{arXiv:#2}

\bibitem[\protect\citeauthoryear{Aumann}{Aumann}{1987}]%
        {aumann1987correlated}
\bibfield{author}{\bibinfo{person}{Robert~J Aumann}.}
  \bibinfo{year}{1987}\natexlab{}.
\newblock \showarticletitle{{Correlated Equilibrium as An Expression of
  Bayesian Rationality}}.
\newblock \bibinfo{journal}{\emph{Econometrica: Journal of the Econometric
  Society}} (\bibinfo{year}{1987}), \bibinfo{pages}{1--18}.
\newblock


\bibitem[\protect\citeauthoryear{Babichenko}{Babichenko}{2016}]%
        {babichenko2016query}
\bibfield{author}{\bibinfo{person}{Yakov Babichenko}.}
  \bibinfo{year}{2016}\natexlab{}.
\newblock \showarticletitle{Query complexity of approximate Nash equilibria}.
\newblock \bibinfo{journal}{\emph{Journal of the ACM (JACM)}}
  \bibinfo{volume}{63}, \bibinfo{number}{4} (\bibinfo{year}{2016}),
  \bibinfo{pages}{1--24}.
\newblock


\bibitem[\protect\citeauthoryear{Barman and Ligett}{Barman and Ligett}{2015}]%
        {barman2015finding}
\bibfield{author}{\bibinfo{person}{Siddharth Barman} {and}
  \bibinfo{person}{Katrina Ligett}.} \bibinfo{year}{2015}\natexlab{}.
\newblock \showarticletitle{{Finding any nontrivial coarse correlated
  equilibrium is hard}}.
\newblock \bibinfo{journal}{\emph{ACM SIGecom Exchanges}} \bibinfo{volume}{14},
  \bibinfo{number}{1} (\bibinfo{year}{2015}), \bibinfo{pages}{76--79}.
\newblock


\bibitem[\protect\citeauthoryear{Berg and Sandholm}{Berg and Sandholm}{2017}]%
        {berg2017exclusion}
\bibfield{author}{\bibinfo{person}{Kimmo Berg} {and} \bibinfo{person}{Tuomas
  Sandholm}.} \bibinfo{year}{2017}\natexlab{}.
\newblock \showarticletitle{{Exclusion Method for Finding Nash Equilibrium in
  Multiplayer Games}}. In \bibinfo{booktitle}{\emph{Proceedings of the AAAI
  Conference on Artificial Intelligence}}, Vol.~\bibinfo{volume}{31}.
\newblock


\bibitem[\protect\citeauthoryear{Berri, Varma, Lasaulce, and Radjef}{Berri
  et~al\mbox{.}}{2016}]%
        {berri2016correlated}
\bibfield{author}{\bibinfo{person}{Sara Berri}, \bibinfo{person}{Vineeth
  Varma}, \bibinfo{person}{Samson Lasaulce}, {and}
  \bibinfo{person}{Mohammed~Said Radjef}.} \bibinfo{year}{2016}\natexlab{}.
\newblock \showarticletitle{{Correlated Equilibria in Wireless Power Control
  Games}}. In \bibinfo{booktitle}{\emph{International Conference on Network
  Games, Control, and Optimization}}. Springer, \bibinfo{pages}{57--68}.
\newblock


\bibitem[\protect\citeauthoryear{Brown}{Brown}{1951}]%
        {brown1951iterative}
\bibfield{author}{\bibinfo{person}{George~W Brown}.}
  \bibinfo{year}{1951}\natexlab{}.
\newblock \showarticletitle{Iterative solution of games by fictitious play}.
\newblock \bibinfo{journal}{\emph{Activity analysis of production and
  allocation}} \bibinfo{volume}{13}, \bibinfo{number}{1}
  (\bibinfo{year}{1951}), \bibinfo{pages}{374--376}.
\newblock


\bibitem[\protect\citeauthoryear{Chandan, Paccagnan, and Marden}{Chandan
  et~al\mbox{.}}{2021}]%
        {chandan2021tractable}
\bibfield{author}{\bibinfo{person}{Rahul Chandan}, \bibinfo{person}{Dario
  Paccagnan}, {and} \bibinfo{person}{Jason~R Marden}.}
  \bibinfo{year}{2021}\natexlab{}.
\newblock \showarticletitle{Tractable Mechanisms for Computing Near-Optimal
  Utility Functions}. In \bibinfo{booktitle}{\emph{Proceedings of the 20th
  International Conference on Autonomous Agents and MultiAgent Systems}}.
  \bibinfo{pages}{306--313}.
\newblock


\bibitem[\protect\citeauthoryear{Feng, Guruganesh, Liaw, Mehta, and Sethi}{Feng
  et~al\mbox{.}}{2021}]%
        {feng2021convergence}
\bibfield{author}{\bibinfo{person}{Zhe Feng}, \bibinfo{person}{Guru
  Guruganesh}, \bibinfo{person}{Christopher Liaw}, \bibinfo{person}{Aranyak
  Mehta}, {and} \bibinfo{person}{Abhishek Sethi}.}
  \bibinfo{year}{2021}\natexlab{}.
\newblock \showarticletitle{Convergence analysis of no-regret bidding
  algorithms in repeated auctions}. In \bibinfo{booktitle}{\emph{Proceedings of
  the AAAI Conference on Artificial Intelligence}}, Vol.~\bibinfo{volume}{35}.
  \bibinfo{pages}{5399--5406}.
\newblock


\bibitem[\protect\citeauthoryear{Hart and Mansour}{Hart and Mansour}{2010}]%
        {hart2010long}
\bibfield{author}{\bibinfo{person}{Sergiu Hart} {and} \bibinfo{person}{Yishay
  Mansour}.} \bibinfo{year}{2010}\natexlab{}.
\newblock \showarticletitle{How long to equilibrium? The communication
  complexity of uncoupled equilibrium procedures}.
\newblock \bibinfo{journal}{\emph{Games and Economic Behavior}}
  \bibinfo{volume}{69}, \bibinfo{number}{1} (\bibinfo{year}{2010}),
  \bibinfo{pages}{107--126}.
\newblock


\bibitem[\protect\citeauthoryear{Hart and Mas-Colell}{Hart and
  Mas-Colell}{2000}]%
        {hart2000simple}
\bibfield{author}{\bibinfo{person}{Sergiu Hart} {and} \bibinfo{person}{Andreu
  Mas-Colell}.} \bibinfo{year}{2000}\natexlab{}.
\newblock \showarticletitle{{A Simple Adaptive Procedure Leading to Correlated
  Equilibrium}}.
\newblock \bibinfo{journal}{\emph{Econometrica}} \bibinfo{volume}{68},
  \bibinfo{number}{5} (\bibinfo{year}{2000}), \bibinfo{pages}{1127--1150}.
\newblock


\bibitem[\protect\citeauthoryear{Hart and Mas-Colell}{Hart and
  Mas-Colell}{2015}]%
        {hart2015markets}
\bibfield{author}{\bibinfo{person}{Sergiu Hart} {and} \bibinfo{person}{Andreu
  Mas-Colell}.} \bibinfo{year}{2015}\natexlab{}.
\newblock \showarticletitle{Markets, correlation, and regret-matching}.
\newblock \bibinfo{journal}{\emph{Games and Economic Behavior}}
  \bibinfo{volume}{93} (\bibinfo{year}{2015}), \bibinfo{pages}{42--58}.
\newblock


\bibitem[\protect\citeauthoryear{Konda, Chandan, and Marden}{Konda
  et~al\mbox{.}}{2021}]%
        {konda2021mission}
\bibfield{author}{\bibinfo{person}{Rohit Konda}, \bibinfo{person}{Rahul
  Chandan}, {and} \bibinfo{person}{Jason~R Marden}.}
  \bibinfo{year}{2021}\natexlab{}.
\newblock \showarticletitle{Mission Level Uncertainty in Multi-Agent Resource
  Allocation}. In \bibinfo{booktitle}{\emph{2021 60th IEEE Conference on
  Decision and Control (CDC)}}. IEEE, \bibinfo{pages}{4521--4526}.
\newblock


\bibitem[\protect\citeauthoryear{Koutsoupias and Papadimitriou}{Koutsoupias and
  Papadimitriou}{1999}]%
        {koutsoupias1999worst}
\bibfield{author}{\bibinfo{person}{Elias Koutsoupias} {and}
  \bibinfo{person}{Christos Papadimitriou}.} \bibinfo{year}{1999}\natexlab{}.
\newblock \showarticletitle{Worst-case equilibria}. In
  \bibinfo{booktitle}{\emph{Annual symposium on theoretical aspects of computer
  science}}. Springer, \bibinfo{pages}{404--413}.
\newblock


\bibitem[\protect\citeauthoryear{Li, Zhao, and Zhu}{Li et~al\mbox{.}}{2022}]%
        {li2022role}
\bibfield{author}{\bibinfo{person}{Tao Li}, \bibinfo{person}{Yuhan Zhao}, {and}
  \bibinfo{person}{Quanyan Zhu}.} \bibinfo{year}{2022}\natexlab{}.
\newblock \showarticletitle{The role of information structures in
  game-theoretic multi-agent learning}.
\newblock \bibinfo{journal}{\emph{Annual Reviews in Control}}
  (\bibinfo{year}{2022}).
\newblock


\bibitem[\protect\citeauthoryear{Marden, Arslan, and Shamma}{Marden
  et~al\mbox{.}}{2007}]%
        {marden2007regret}
\bibfield{author}{\bibinfo{person}{Jason~R Marden}, \bibinfo{person}{G{\"u}rdal
  Arslan}, {and} \bibinfo{person}{Jeff~S Shamma}.}
  \bibinfo{year}{2007}\natexlab{}.
\newblock \showarticletitle{Regret based dynamics: convergence in weakly
  acyclic games}. In \bibinfo{booktitle}{\emph{Proceedings of the 6th
  international joint conference on Autonomous agents and multiagent systems}}.
  \bibinfo{pages}{1--8}.
\newblock


\bibitem[\protect\citeauthoryear{Marden, Arslan, and Shamma}{Marden
  et~al\mbox{.}}{2009}]%
        {marden2009cooperative}
\bibfield{author}{\bibinfo{person}{Jason~R Marden}, \bibinfo{person}{G{\"u}rdal
  Arslan}, {and} \bibinfo{person}{Jeff~S Shamma}.}
  \bibinfo{year}{2009}\natexlab{}.
\newblock \showarticletitle{Cooperative control and potential games}.
\newblock \bibinfo{journal}{\emph{IEEE Transactions on Systems, Man, and
  Cybernetics, Part B (Cybernetics)}} \bibinfo{volume}{39}, \bibinfo{number}{6}
  (\bibinfo{year}{2009}), \bibinfo{pages}{1393--1407}.
\newblock


\bibitem[\protect\citeauthoryear{Marden and Shamma}{Marden and Shamma}{2015}]%
        {marden2015game}
\bibfield{author}{\bibinfo{person}{Jason~R Marden} {and}
  \bibinfo{person}{Jeff~S Shamma}.} \bibinfo{year}{2015}\natexlab{}.
\newblock \showarticletitle{Game theory and distributed control}.
\newblock In \bibinfo{booktitle}{\emph{Handbook of game theory with economic
  applications}}. Vol.~\bibinfo{volume}{4}. \bibinfo{publisher}{Elsevier},
  \bibinfo{pages}{861--899}.
\newblock


\bibitem[\protect\citeauthoryear{Mihailescu, Nguyen, and Webb}{Mihailescu
  et~al\mbox{.}}{2015}]%
        {7348945}
\bibfield{author}{\bibinfo{person}{Marian Mihailescu}, \bibinfo{person}{Hung
  Nguyen}, {and} \bibinfo{person}{Michael~R. Webb}.}
  \bibinfo{year}{2015}\natexlab{}.
\newblock \showarticletitle{Enhancing wireless communications with software
  defined networking}. In \bibinfo{booktitle}{\emph{2015 Military
  Communications and Information Systems Conference (MilCIS)}}.
  \bibinfo{pages}{1--6}.
\newblock
\urldef\tempurl%
\url{https://doi.org/10.1109/MilCIS.2015.7348945}
\showDOI{\tempurl}


\bibitem[\protect\citeauthoryear{Monderer and Shapley}{Monderer and
  Shapley}{1996}]%
        {monderer1996potential}
\bibfield{author}{\bibinfo{person}{Dov Monderer} {and} \bibinfo{person}{Lloyd~S
  Shapley}.} \bibinfo{year}{1996}\natexlab{}.
\newblock \showarticletitle{Potential games}.
\newblock \bibinfo{journal}{\emph{Games and economic behavior}}
  \bibinfo{volume}{14}, \bibinfo{number}{1} (\bibinfo{year}{1996}),
  \bibinfo{pages}{124--143}.
\newblock


\bibitem[\protect\citeauthoryear{Nguyen, Nguyen, and White}{Nguyen
  et~al\mbox{.}}{2016}]%
        {7878785}
\bibfield{author}{\bibinfo{person}{Duong~D. Nguyen}, \bibinfo{person}{Hung~X.
  Nguyen}, {and} \bibinfo{person}{Langford~B. White}.}
  \bibinfo{year}{2016}\natexlab{}.
\newblock \showarticletitle{Performance of adaptive RAT selection algorithms in
  5G heterogeneous wireless networks}. In \bibinfo{booktitle}{\emph{2016 26th
  International Telecommunication Networks and Applications Conference
  (ITNAC)}}. \bibinfo{pages}{70--75}.
\newblock
\urldef\tempurl%
\url{https://doi.org/10.1109/ATNAC.2016.7878785}
\showDOI{\tempurl}


\bibitem[\protect\citeauthoryear{Nguyen, Nguyen, and White}{Nguyen
  et~al\mbox{.}}{2017}]%
        {nguyen2017reinforcement}
\bibfield{author}{\bibinfo{person}{Duong~D Nguyen}, \bibinfo{person}{Hung~X
  Nguyen}, {and} \bibinfo{person}{Langford~B White}.}
  \bibinfo{year}{2017}\natexlab{}.
\newblock \showarticletitle{{Reinforcement Learning with Network-Assisted
  Feedback for Heterogeneous RAT Selection}}.
\newblock \bibinfo{journal}{\emph{IEEE Transactions on Wireless
  Communications}} \bibinfo{volume}{16}, \bibinfo{number}{9}
  (\bibinfo{year}{2017}), \bibinfo{pages}{6062--6076}.
\newblock


\bibitem[\protect\citeauthoryear{Nguyen, Nguyen, and White}{Nguyen
  et~al\mbox{.}}{2018}]%
        {8489869}
\bibfield{author}{\bibinfo{person}{Duong~D. Nguyen}, \bibinfo{person}{Hung~X.
  Nguyen}, {and} \bibinfo{person}{Langford~B. White}.}
  \bibinfo{year}{2018}\natexlab{}.
\newblock \showarticletitle{Evaluating Performance of RAT Selection Algorithms
  for 5G Hetnets}.
\newblock \bibinfo{journal}{\emph{IEEE Access}}  \bibinfo{volume}{6}
  (\bibinfo{year}{2018}), \bibinfo{pages}{61212--61222}.
\newblock
\urldef\tempurl%
\url{https://doi.org/10.1109/ACCESS.2018.2875469}
\showDOI{\tempurl}


\bibitem[\protect\citeauthoryear{Nguyen, Rajagopalan, Kim, and Lim}{Nguyen
  et~al\mbox{.}}{2019}]%
        {nguyen2019adaptive}
\bibfield{author}{\bibinfo{person}{Duong~D Nguyen}, \bibinfo{person}{Arvind
  Rajagopalan}, \bibinfo{person}{Jijoong Kim}, {and}
  \bibinfo{person}{Cheng~Chew Lim}.} \bibinfo{year}{2019}\natexlab{}.
\newblock \showarticletitle{Adaptive regret minimization for learning complex
  team-based tactics}.
\newblock \bibinfo{journal}{\emph{IEEE Access}}  \bibinfo{volume}{7}
  (\bibinfo{year}{2019}), \bibinfo{pages}{103019--103030}.
\newblock


\bibitem[\protect\citeauthoryear{Paccagnan, Chandan, and Marden}{Paccagnan
  et~al\mbox{.}}{2022}]%
        {paccagnan2022utility}
\bibfield{author}{\bibinfo{person}{Dario Paccagnan}, \bibinfo{person}{Rahul
  Chandan}, {and} \bibinfo{person}{Jason~R Marden}.}
  \bibinfo{year}{2022}\natexlab{}.
\newblock \showarticletitle{Utility and mechanism design in multi-agent
  systems: An overview}.
\newblock \bibinfo{journal}{\emph{Annual Reviews in Control}}
  (\bibinfo{year}{2022}).
\newblock


\bibitem[\protect\citeauthoryear{Qu, Brown, and Li}{Qu et~al\mbox{.}}{2019}]%
        {qu2019distributed}
\bibfield{author}{\bibinfo{person}{Guannan Qu}, \bibinfo{person}{Dave Brown},
  {and} \bibinfo{person}{Na Li}.} \bibinfo{year}{2019}\natexlab{}.
\newblock \showarticletitle{Distributed greedy algorithm for multi-agent task
  assignment problem with submodular utility functions}.
\newblock \bibinfo{journal}{\emph{Automatica}}  \bibinfo{volume}{105}
  (\bibinfo{year}{2019}), \bibinfo{pages}{206--215}.
\newblock


\bibitem[\protect\citeauthoryear{Raducha and San~Miguel}{Raducha and
  San~Miguel}{2022}]%
        {raducha2022coordination}
\bibfield{author}{\bibinfo{person}{Tomasz Raducha} {and} \bibinfo{person}{Maxi
  San~Miguel}.} \bibinfo{year}{2022}\natexlab{}.
\newblock \showarticletitle{Coordination and equilibrium selection in games:
  the role of local effects}.
\newblock \bibinfo{journal}{\emph{Scientific reports}} \bibinfo{volume}{12},
  \bibinfo{number}{1} (\bibinfo{year}{2022}), \bibinfo{pages}{1--16}.
\newblock


\bibitem[\protect\citeauthoryear{Roughgarden}{Roughgarden}{2015}]%
        {roughgarden2015intrinsic}
\bibfield{author}{\bibinfo{person}{Tim Roughgarden}.}
  \bibinfo{year}{2015}\natexlab{}.
\newblock \showarticletitle{Intrinsic robustness of the price of anarchy}.
\newblock \bibinfo{journal}{\emph{Journal of the ACM (JACM)}}
  \bibinfo{volume}{62}, \bibinfo{number}{5} (\bibinfo{year}{2015}),
  \bibinfo{pages}{1--42}.
\newblock


\bibitem[\protect\citeauthoryear{Saran and Serrano}{Saran and Serrano}{2012}]%
        {saran2012regret}
\bibfield{author}{\bibinfo{person}{Rene Saran} {and} \bibinfo{person}{Roberto
  Serrano}.} \bibinfo{year}{2012}\natexlab{}.
\newblock \showarticletitle{Regret matching with finite memory}.
\newblock \bibinfo{journal}{\emph{Dynamic Games and Applications}}
  \bibinfo{volume}{2}, \bibinfo{number}{1} (\bibinfo{year}{2012}),
  \bibinfo{pages}{160--175}.
\newblock


\bibitem[\protect\citeauthoryear{Vetta}{Vetta}{2002}]%
        {vetta2002nash}
\bibfield{author}{\bibinfo{person}{Adrian Vetta}.}
  \bibinfo{year}{2002}\natexlab{}.
\newblock \showarticletitle{Nash equilibria in competitive societies, with
  applications to facility location, traffic routing and auctions}. In
  \bibinfo{booktitle}{\emph{The 43rd Annual IEEE Symposium on Foundations of
  Computer Science, 2002. Proceedings.}} IEEE, \bibinfo{pages}{416--425}.
\newblock


\bibitem[\protect\citeauthoryear{Zhang, Lerer, and Brown}{Zhang
  et~al\mbox{.}}{2022}]%
        {zhang2022equilibrium}
\bibfield{author}{\bibinfo{person}{Hugh Zhang}, \bibinfo{person}{Adam Lerer},
  {and} \bibinfo{person}{Noam Brown}.} \bibinfo{year}{2022}\natexlab{}.
\newblock \showarticletitle{Equilibrium Finding in Normal-Form Games via Greedy
  Regret Minimization}. In \bibinfo{booktitle}{\emph{Proceedings of the AAAI
  Conference on Artificial Intelligence}}, Vol.~\bibinfo{volume}{36}.
  \bibinfo{pages}{9484--9492}.
\newblock


\end{thebibliography}

%%%%%%%%%%%%%%%%%%%%%%%%%%%%%%%%%%%%%%%%%%%%%%%%%%%%%%%%%%%%%%%%%%%%%%%%
% \newpage \clearpage
% \input{appendix}
\appendix
\section{Appendix}

\subsection{Proof of Lemma~\ref{Existence_PSNE}}
\begin{proof}
Every finite coordination game in which the global objective function aligned with the local utility functions of the players, that is, satisfies the property as in~\eqref{eq:potential_function}, is a generalised ordinal potential game~\cite{marden2009cooperative}.
Let $\phi$ be a potential function of a coordination game $\Gc$. Then the equilibrium set of $\Gc$ corresponds to the set of local maxima of $\phi$. That is, an action profile $a =(a_i,a_{-i})$ is a NE point for $\Gc$ if an only if for every $i\in\Nc$,
$$\phi(a) \geq \phi(a^{\prime}_i, a_{-i}),\ \forall a^{\prime}_i\in \Ac_i \ .$$
%Consequently, if $\phi$ admits a maximal value in $\Ac$, which is true by definition for a finite set $\Ac$, then $\phi$ processes a pure strategy Nash equilibrium.

Consider $a^*=(a_i^*,a_{-i}^*)\in\Ac$ for which $\phi(a^*)$ is maximal (which is true by definition for a finite set $\Ac$), then for any $a^{\prime}=(a_i^{\prime},a_{-i})$:
\begin{multline*}
\phi(a_i^*,a_{-i}^*) > \phi(a^{\prime}_i,a_{-i}) 
\Leftrightarrow u_i(a_i^*,a_{-i}^*) > u(a^{\prime}_i,a_{-i})\ .
\end{multline*}
Hence, the game possesses a pure strategy NE.
\end{proof}

\subsection{Experiments with Stag Hunt Game}
We also run SORM and LURM algorithms for a Stag Hunt game and find similar results to the experiments for the $2\times2$ resource allocation game described in the section of Experimental Results. The Stag Hunt game (illustrated below in Figure~\ref{fig:StagHuntExample}) is a typical example for coordination game, wherein two individuals go out on a hunt. Each can individually choose to hunt a stag or a hare, and must choose a decision at the same time, without knowing the choice of the other. If an individual hunts a stag, they must have the cooperation of their partner in order to succeed. An individual can get a hare by themselves, but a hare is worth less than a stag.

\begin{figure}[!h]
 \centering
 \includegraphics[width=0.9\linewidth]{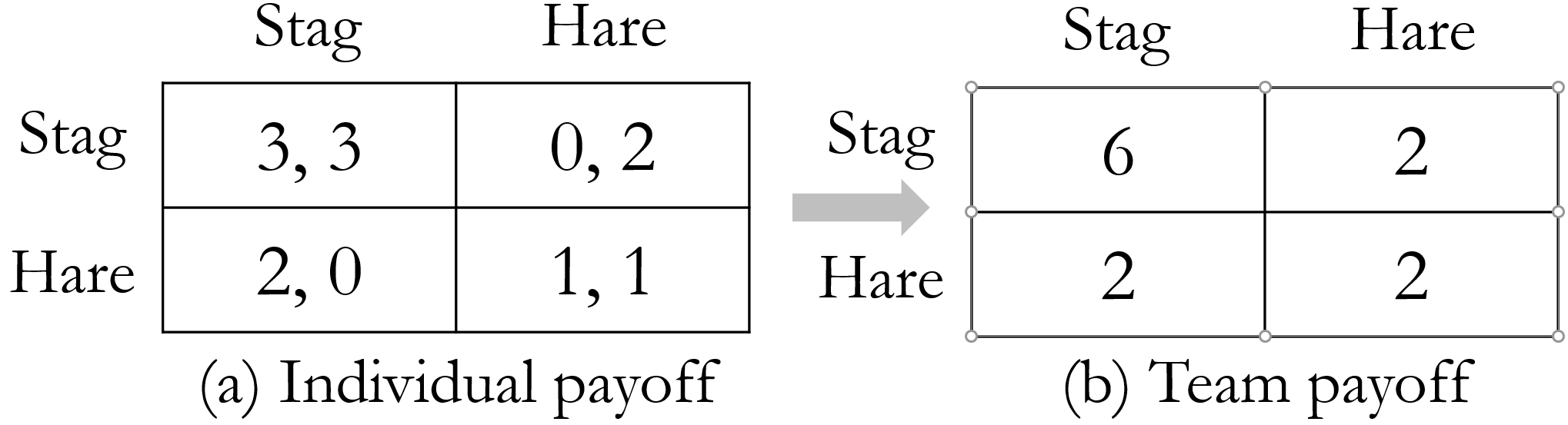}
 \caption{Payoff matrix of a simple $2\times2$ Stag Hunt game.}
 \label{fig:StagHuntExample}
\end{figure}

In the Stag Hunt game, only one action combination is socially optimal, which is when they both hunt stag, however there exist two possible NE outcomes: (hare, hare) and (stag, stag). With the later outcome, both players receive a payoff of $3$, which yields a largest social welfare of $6$ for the game. In this case, we cannot switch to any other outcome and make at least one party better off without making anyone worse off. 

Note that although (stag, stag) is the best solution for all players, neither of them would unilaterally deviate from the suboptimal NE (hare, hare) to improve its payoff since individual payoff would drop from $1$ to $0$. Because of that, players can often get stuck in an inferior equilibrium. Thus, if the players could jointly agree on what to do, then it makes sense for them to choose a payoff-maximising strategy at the (stag, stag) solution. This is a commons dilemma situation in many real-world games (i.e., businesses), in which a large number of players have to make their decisions at the same time. Each decision is often made upon a limited set of information and it is not always possible for them to coordinate their decisions or jointly plan strategy before playing the game. 

\begin{figure}[!t]
  \centering
  \begin{subfigure}[b]{0.493\linewidth}
    \includegraphics[width=\linewidth]{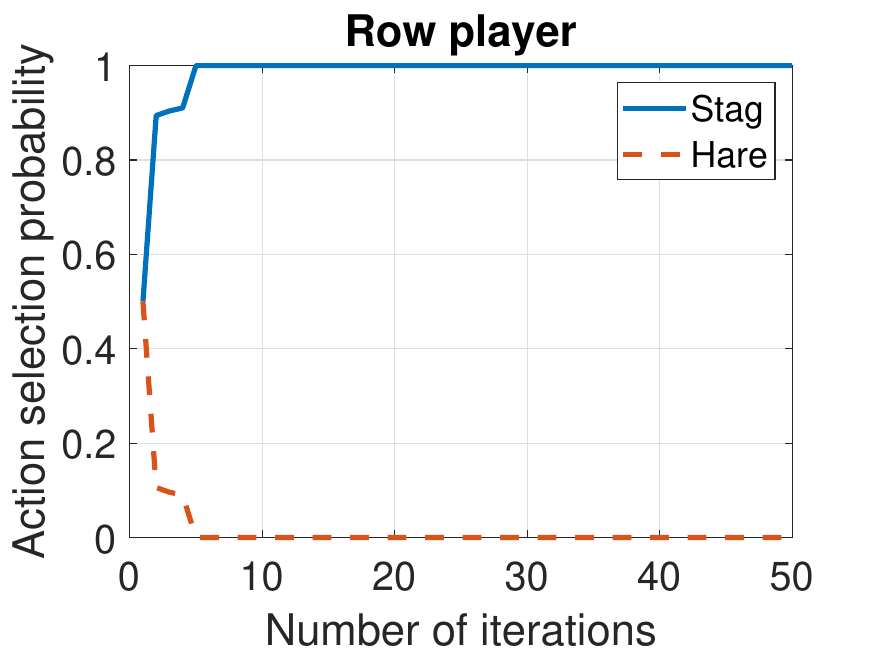}
    % \caption{Coffee.}
  \end{subfigure}
  \begin{subfigure}[b]{0.493\linewidth}
    \includegraphics[width=\linewidth]{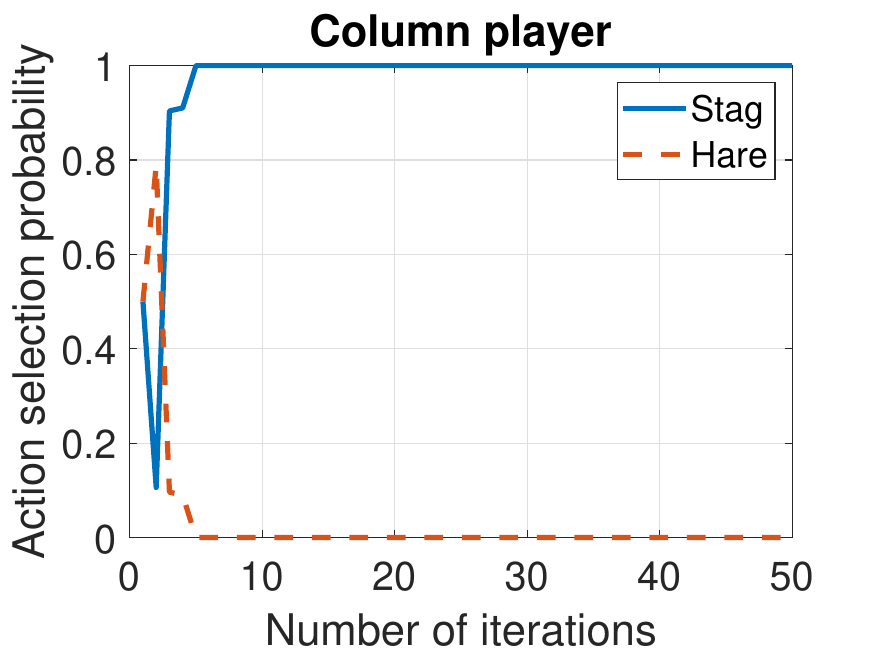}
    % \caption{More coffee.}
  \end{subfigure}
  \caption{Evolution of action selection probabilities of all players versus iterations in the Stag-Hunt game.}
  \label{fig:StagHare}
\end{figure}

% We first demonstrate the performance of our approach in achieving the social optimum outcome of the well-known $2\times2$ Stag Hare game illustrated previously in Figure~\ref{fig:example}. 
Figure~\ref{fig:StagHare} demonstrated the performance of our method on the Stag Hunt game. In our simulation, both players running our proposed algorithm always converge to the same joint action (stag, stag), which is more socially desirable than the other PSNE (hare, hare). Meanwhile, the result of running the standard RM (LURM) randomly converges to one of the two Nash points.

% \begin{figure}[!h]
%  \centering
%  \includegraphics[width=0.5\linewidth]{figures/StagHare_rowplayer}
%  \caption{stag hunt game.}
%  \label{fig:example}
% \end{figure}

%%%%%%%%%%%%%%%%%%%%%%%%%%%%%%%%%%%%%%%%%%%%%%%%%%%%%%%%%%%%%%%%%%%%%%%%

\end{document}